\documentclass[9pt]{article}
 \date{}
 \title{Exact sequence between real and complex bivariant $\K$ theories and application to the $\ZZ$ pairing}
 \author{Samuel Guerin}
\usepackage{lmodern} %
\usepackage[utf8]{inputenc} %
\usepackage[T1]{fontenc} %

\usepackage[x11names]{xcolor}
\usepackage{pstricks}
\usepackage{amsmath}
\usepackage{amsfonts}%
\usepackage{amsthm}%
\usepackage{amssymb}%
\usepackage{mathtools}%
\usepackage{enumitem}%
\usepackage{bm}%
\usepackage{slashed}%
\usepackage{mathrsfs}%
\usepackage{mathabx}%
\usepackage{booktabs}%
\usepackage{geometry}
\usepackage{tikz}
\usetikzlibrary{arrows,shapes}
\usetikzlibrary{tikzmark}
\usepackage{tikz-cd}
\usepackage{accents}
\usepackage[german=quotes,french=quotes]{csquotes}
\usepackage[style=alphabetic,
    sortlocale=en_US,
    natbib=false,
    url=false,
    doi=false,
    eprint=true,
    backend=bibtex]{biblatex}
\usepackage[linktocpage=true]{hyperref}%
\hypersetup{
    colorlinks=false,
}
\usepackage{cleveref}%

\DeclareMathOperator{\tr}{Tr}

\DeclareMathOperator{\id}{id}
\DeclareMathOperator{\re}{Re}

\def\epsilon{\varepsilon}

\def\hom{\mathit{\rm Hom}}

\def\Z{\mathbb{Z}}
\def\ZZ{\mathbb{Z}_2}
\def\N{\mathbb{N}}
\def\C{\mathbb{C}}
\def\H{\mathbb{H}}
\def\Q{\mathbb{Q}}
\def\R{\mathbb{R}}
\def\Rpq{\mathbb{R}^{p, q}}
\def\T{\mathbb{T}}

\let\kkfont\operatorname
\def\KK{\kkfont{KK}}
\def\KKO{\kkfont{KKO}}
\def\KKU{\kkfont{KKU}}

\def\K{\kkfont{K}}
\def\KO{\kkfont{KO}}
\def\KU{\kkfont{KU}}

\def\KOK{\KO_{\kkfont{K}}}
\def\kok{\KO^{\kkfont{K}}}
\def\HC{\kkfont{HC}}
\def\HCO{\kkfont{HCO}}

\def\G{\mathcal{G}}
\def\O{\operatorname{O}}
\def\SO{\operatorname{SO}}
\def\U{\operatorname{U}}
\def\S{\mathfrak{S}}
\def\m{\mathfrak{m}}
\def\p{\mathfrak{p}}
\def\x{\textbf{x}}
\def\M{M}
\def\simarrow{\xrightarrow{\sim}}
\def\incl{\hookrightarrow}

\def\fun{\mathrm{C}}
\def\func{\mathrm{C}^0}
\def\dirac{\slashed{\mathrm{D}}}
\def\dualdirac{\slashed{\mathrm{x}}}
\newcommand{\clif}[1][]{C \ell^{#1}}
\newcommand{\cclif}[1][]{C\ell^{#1}_\C}
\def\clifpq{\clif[p,q]}
\def\cclifpq{\cclif[p,q]}

\def\hilb{\mathcal{H}}
\let\opfont\mathit
\def\compact{\opfont K}
\def\summable{\mathcal L}
\def\bounded{\mathscr L}
\def\calkin{\opfont{Cal}}

\def\gotimes{\hat{\otimes}}
\def\endo{\mathscr{L}}

\def\totxt{\xrightarrow}
\def\lambdarpq{\Lambda^\star \R^{p+q}}
\def\cstar{\mathcal{C}^\star \text{algebra}}

\def\cstarr{\text{real } \cstar}
\def\cstars{\mathcal{C}^\star\text{algebras}}

\def\cstarrs{\text{real } \cstars}
\def\proj{\mathbb P}

\def\bardot{\overbar{\hspace{2pt} \boldsymbol{\cdot} \vphantom{\otimes} \hspace{2pt}}}

\def\ct{\mathcal C}
\def\sct{\mathcal H}
\def\gct{\mathcal G}
\def\ori{\mathrm o}
\newcommand{\norm}[1]{  \Vert #1 \Vert }
\newcommand{\cone}[1]{  C_{#1} }
\newcommand{\shm}{-}

\newcommand{\overbar}[1]{\mkern 1.5mu\overline{\mkern-1.5mu#1\mkern-1.5mu}\mkern 1.5mu}
\newcommand{\intercc}[1]{\left[ #1 \right]}
\newcommand{\interco}[1]{\left[ #1 \right[}

\newcommand{\interoo}[1]{\left] #1 \right[}

\theoremstyle{definition}
\newtheorem{defi}{Definition}
\theoremstyle{plain}
\newtheorem{theo}{Theorem}
\newtheorem{prop}{Proposition}
\newtheorem{cor}{Corollary}
\newtheorem{lemma}{Lemma}
\theoremstyle{remark}
\newtheorem{remark}{Remark}
\theoremstyle{plain}

\bibliography{Documents}

\begin{document}

\maketitle{}
\setcounter{tocdepth}{2}
\paragraph{Abstract}
We give some formulas for the $\ZZ$ pairing in $\KO$ theory using a long exact sequence for bivariant $\K$ theory which links real and complex theories.
This is discussed under the framework of real structures given by antilinear operators verifying some symmetries. 
Topological phases protected by time reversal symmetry from condensed matter physics will be discussed. 
\tableofcontents

\section{Introduction}

In condensed matter physics, the last forty years have seen the rise and development of the topic of topological phases of matter.
First appeared when shifting the band theory in solid state physics to a deeper analysis of global properties of phase functions of electrons, it includes now photonic or mechanical systems.
These materials exhibit interesting properties such as an insulator behavior in the bulk and a non zero conductivity in the boundary.
As an example, Quantum Hall effect was first recognized as in incarnation of the topology of the Brillouin zone, the Fourier space  of a periodic crystal.
For a rational flux quantum, Thouless Kohmoto and Nightingale~\cite{tkn82} first link the conductance in the Quantum Hall Effect with the first Chern number of the spaces spanned by filled bands of electrons.
These phases are then characterized by topological numbers.
Refinement that model disorder are given as a crossed product algebra  $\fun(\Omega)\rtimes \Z^d$ for $\omega$ a space with a $\Z^d$ action representing disorder.
Moreover, a non vanishing magnetic field twists the crossed product with $\Z^d$.
Those two observations make non-commutative geometry inevitable for such a problem as first described by Bellissard et al.~\cite{bvs94}.
Other models based on non-commutative algebras have been proposed in order to integrate disorder.
Examples of such models include
Kubota~ \cite{k17b} interpretation based on reduced Roe algebras, Ewert and Meyer~\cite{em18a} with uniform Roe algebras.
One may consult the last paper for an interesting comparison of such models.

Those systems are described by a first quantized Hamiltonian $H$ acting on the phase space of a single electron and by the Fermi energy.
The latter is a real number characterizing the statistic of electrons: at zero temperature all electrons occupy all energies below it.
With conditions such as $\Z^d$ translational invariance, growth conditions on the action of the Hamiltonian between distant sites or the Fermi energy lying in a spectral gap of $H$, one may consider the class of the spectral projection of the Hamiltonian below the Fermi energy in the $\K$ theory group of one of the algebra above.
Topological numbers, such as Chern numbers are then produced by pairing this $\K$ class with cyclic homology.a
This is then linked to physical quantities.

Since the work of Kane and Mele, more and more interest is devoted to systems with extra invariance.
Symmetries protect from arbitrary disorder and allow to distinguish new phases.  
Time reversal invariant systems were the first predicted by Kane-Mele in 2D~\cite{km05} and measured two years later~\cite{kwb07}.
It leads to interesting new phenomena such as Quantum Spin Hall Effect~\cite{km05a,kbm08}.
Some time reversal invariant phases are characterized by a $\ZZ$ topological number. 
Combining this symmetry with particle-hole symmetry leads to ten symmetry classes described for translation invariant systems in any dimension by Kitaev~\cite{k09}.
Henceforth, other types of symmetry protected phases have been considered with crystal point symmetry groups; see for example~\cite{f11}.
See also~\cite{fm13} for a discussion on gaped systems protected by both their topology and by general symmetries.

Time reversal symmetry is incorporated in our model via an antiunitary operator acting on the phase space of an electron and therefore on our algebra.
The theory of real algebras and their $\K$ theory then come naturally into play.
The $\ZZ$ topological numbers can be seen as a pairing in real $\K$ theory between  the class of a projector and a $\K$ homology element.
Real $\K$ theory was developed jointly as its complex analogue.
Work has been done on the specificity of the real theory, including the eightfold Bott periodicity~\cite{abs64, k68} and its extension to real Banach algebras~\cite{w66}, for the index theory of real elliptic pseudo-differential operators~\cite{a66, as69, k70, as71} and its applications to the problem of the existence of metrics of positive scalar curvature on smooth manifolds~\cite{l63, s92}, study of Baum-Connes conjecture in the real setting~\cite{k01,s04,bk04} or general properties of bivariant $\KKO$ functor~\cite{b03, br09}.

We place ourselves in the bivariant theory of Kasparov~\cite{k81}.
This theory will be briefly recalled in the first sections.
We will focus on the $\ZZ$ pairing between $\KO$ theory and $\KO$ homology.
This will be studied through an exact sequence relating real and complex $\K$ theory.
The existence of such a sequence was first shown by Atiyah in the monovariant commutative case in~\cite{a66}, another proof is given by Schick in~\cite{s04}, extends the exact sequence to the bivariant non-commutative case.
We will unravel Atiyah's approach in the bivariant setting to give applications to our problem.
Then, we recall how to express each of the eight $\KKO$ groups via the presence of antilinear symmetries of a given type.
The exact sequence will be then described in this framework.
We then give formulas for the $\ZZ$ pairing between classes that vanishes in the exact sequence, using $\Z$ valued formulas.
Such formulas will rely on the data given by an homotopy representing the vanishing of the considered $\KO$ class as in~\cite{k16}.
Finally, we discuss symmetries from the Wigner point of view, linking symmetries for the eight $\KKO$ groups to the periodic table of Kitaev and with time reversal invariant topological insulators.
We will not go into details on the theory of topological insulators, focusing on general considerations on symmetries and reality.
We will not discuss the chosen model algebra.

We wish to thank professor Denis Perrot for his insights and support, to thank professor Johannes Kellendonk for his his presentation of the subject of topological insulators and for sharing his work with us. We also thank professor Ralf Meyer and Schulz-Baldes, Varghese Mathai for interesting and valuable discussions.
 
\section{Real \texorpdfstring{$\KO$}{KO} theory}

\subsection{Real \texorpdfstring{$\K$}{K} theory for graded \texorpdfstring{$\cstarrs$}{real C* algebras}}

In the domain of non-commutative geometry we study here $\cstarrs$.
Recall the definition of these objects, knowing the one for complex $\cstars$:
\begin{defi}
  A $\cstarr$ is Banach $*-$algebra over the field of real numbers such that $A \otimes \C$ is a complex $\cstar$.
  Morphisms of $\cstarrs$ are given by $*-$algebra morphisms.
\end{defi}
For such a unital algebra $\KO$ theory classes are given by pairs of isomorphism classes of projectors in some $\M_n(A)$ considered as formal differences of such projectors.
For non unital $A$ we consider pairs of projectors in the algebra of matrices over the point unitalization $A^+$ such that the difference belongs to $A$.

We then define highers $\KO$ groups considering suspension algebras  $S^nA = \func(\interoo{0, 1}^n) \otimes A$ consisting of continuous functions from the closed interval $\R^n$ to $A$ vanishing at infinity.
 Defining $\KO_n(A)$ as $\KO_0(S^n A)$ we obtain a family of covariant functors indexed by $\N$ stable, homotopy invariant and taking short exact sequences of algebras to long exact sequences of abelian groups. 

Such a theory is extended for non-commutative Banach algebras using the work of Wood on the homotopy of set of unitaries of such algebras~\cite{w66}.
Following them:

\begin{defi}\label{defkaroubi}
  Let $A$ a unital graded $\cstarr$.
  Define $\kok_{p,q}(A)$ as the set described by tuples $(C,  s_1, s_2)$ where $C$ is a $\clif[q,p]-$action on some $A^n$, $A$ linear on the right, $s_1$ and $s_2$ acts $A-$ linearly on $A^n$  as self-adjoint unitaries, and anti-commuting  with the Clifford action.

If $A$ is not unital we change $A$ for its point unitalization $A^+$ and assume that $s_1 = s_2 \mod A$.
  
When then mod out by isomorphic triples and triples for which $s_1$ and $s_2$ can be continuously deformed to one another.
The direct sum then gives this set the structure of an abelian group.
\end{defi}

In fact, such functor could have been defined for graded algebras.
Such approach was taken by Van Daele~\cite{v88a, v88} for non-commutative graded Banach algebras using the work of Wood in its full generality.
Graded algebras are defined as algebras together with an involution.
Recall that say that $a$ is \emph{homogeneous} and that $\partial_a = 0$ (resp. 1).%
For two homogeneous elements $x$ and $y$ in such an algebra, denote the graded tensor product $[x, y] = xy - (-1)^{\partial_x \partial_y}yx$.
We will freely use the graded tensor products $A \gotimes B$ for two graded algebras and corresponding definition for graded modules on such algebras.
Recall just that this construction implies that if both $A$ and $B$ are unital elements any elements $a \gotimes 1$ graded commutes with $1 \gotimes b$ in the sense that the graded commutator vanishes. 
We refer to~\cite[Part~I]{def99} for a detailed exposition of the symmetric monoidal structure associated to such a tensor product.

Asking every generators of a Clifford algebra to be odd, we give $\clifpq$ the structure of a graded algebra.
We recover Karoubi's functors when specializing Van Daele definition for $A \otimes \clif[p+1, q]$.
Note the index of the Clifford algebra.
This will be detailed bellow.
The Bott periodicity can be expressed in such a theory as en equivalence between suspending an algebra and making the graded tensor product with $\clif[0,1]$.

A projector $p \in \M_n(A)$ corresponds to the pair of unitaries $(1-2p, 1)$ and no Clifford action.
A unitary $u \in \M_n(A)$ gives a class in this framework when considering the triple $\left(C,
\left( \begin{smallmatrix} 0 & u^* \\ u & 0 \end{smallmatrix}\right),
\left( \begin{smallmatrix} 0 & 1 \\ 1 & 0 \end{smallmatrix}\right)\right)$
and with Clifford action $C$ given by $\left( \begin{smallmatrix} 0 & -1 \\ 1 & 0 \end{smallmatrix}\right)$.

When considering all commutative $\cstarrs$ we recover by Serre-Swan theorem the $\KO$ theory for real spaces defined by Atiyah~\cite{a66}.
Such algebras are given for $X$ a locally compact space and $\tau : X \to X$ proper of order 2 by
\[ \func(X, \tau) = \{ f: X \to \C, \text{ continuous and vanishing at } \infty, \forall (x,y) \in X, f(\tau(x)) = \overbar{f(x)} \}\]
The pair  $(X, \tau)$ is then called a \emph{real space}.
A $\KO$ group is then generated by stable isomorphism classes of complex vector bundles together with an antilinear map which, on the base space, coincides with the involution of the real space.

Following Atiyah we denote
$\Rpq$ the real space $\R^{p}\times \R^q$ given with the map $\id \times -\id$ and by $S^{p,q}$ its point compactification.

\subsection{The ring \texorpdfstring{$\KO_*$}{KO*}}
The graded abelian group $\KO_*(\R)$ consists of stable classes of vector bundles of vitrual rank zero on spheres.
This pictures corresponds to the description of this abelian group as the kernel of $\KO(\fun(S^n)) \to \KO(\R) = \Z$ upon any inclusion of a point into $ S^n$.
It is also characterized as the cokernel of $\Z \to \KO(\fun(S^n)$ obtained via the map sending the whole sphere $S^n$ to a point,.
It is called the \emph{reduced } $\KO$ \emph{group} of the sphere $S^n$ denoted $\tilde \KO (S^n)$.
Using the tensor product of vector bundles, we give this group a natural product.
This ring structure  may also be obtained via a Kasparov product $\KO(S^n\R) \otimes \KO(S^m\R) \to \KO(S^{n+m}\R)$.
Also, for a graded $\cstarr$, the ablian group $\KO_*(A)$ is given by funtoriality a module structure over the ring $\KO_*(\R)$.

In this section we will describe this $\Z-$graded ring $\KO_*$.

 \begin{itemize}
   \item $1 \in \KO_0$  given by a projector of rank $1$ in $\M_n(\R)$.
   It is of infinite order and generates $\KO_0$.
   \item $\eta \in \KO_1$  given by $[p] - [1]$ where $p$ is the Möbius strip on $S^1$, the compactification of $\interoo{0,1}$.
   It can be identified as the tautological vector bundle of the projective line $\proj_\R^1 \simeq S^1$.
   It is of order 2 and generates $\KO_1$.
   When seen as a unitary $\eta$ is represented by $-\id$ in $\O(1)$.
   Indeed the Möbius band is obtain by gluing the trivial $\R$ at the two endpoints of $\intercc{0,1}$ by $-\id$.
   The square $\eta^2$ generates $\KO_2$.
   This elements can be related as before to the tautological bundle on the complex projective line $\proj_\C^1 \simeq S^2$.
   The element $\eta^3$ is zero, and $\KO_3 = 0$.
   \item $\delta \in \KO_4$ is given by $[\mathcal{O}_\H] - [4]$ where $\mathcal O_\H$ is a vector bundle of rank $4$ on $S^4$ given as a real vector bundle by the tautological bundle of the quaternionic projective line $\proj_\H^1 \simeq S^4$.
   This element of infinite order generates $\KO_4$.
 \item The three groups $\KO_5, \KO_6$ and $\KO_7$ are trivial.
   For example we have $\delta \eta = 0$.
   \item $\gamma \in \KO_8$ can be given in the same fashion as before as the difference between a rank 8 real vector bundle on $S^8$ minus the rank 8 trivial bundle.
     It is of infinite order, generates $\KO_8$ and induces the real Bott periodicity $\KO_k(A) \simarrow \KO_{k+8}(A)$.
     
 \end{itemize}

 Using the isomorphism induced by the product with $\gamma \in \KO_8$, the functor $\KO_n$ can be defined for negative $n$, obtaining an 8-periodic family of covariant functors.

Bott and Shapiro were the first in~\cite{abs64} to understood this 8-periodicity through the scope of Clifford algebras.
Such algebras are defined for any pair of non negative integers $(p,q)$ by $p+q$ generators, $p$ of which squares to $1$ and the last $q$ of them squares to $-1$ each one anti-commuting with one an other.
Asking that every generator is unitary we obtain a unique $\cstarr$ structure on $\clifpq$.
Karoubi defined functors $\KOK^{p,q}$ from compact spaces to abelian groups to formulate Bott and Shapiro's insight~\cite{k68}.

\begin{remark}
   We take here the opposite convention of Karoubi~\cite{k68, k08a}, where $\clifpq$ is built out of $p$ generators of square $-\id$
\end{remark}

\subsection{\texorpdfstring{$\KO$}{KO} homology and \texorpdfstring{$\KKO$}{KKO} theory}

$\KO$ homology is the theory dual to $\KO$ theory.
Cocycles, as generalization of 0th order elliptic pseudo-differential operators, are given by tuples $( \hilb, \epsilon, \phi, F)$ where $ \hilb$ is a real Hilbert space together with a grading $\epsilon : \hilb \to \hilb$ i.e. an involutory unitary.
This grading gives the bounded operators $\bounded(\hilb)$ on $\hilb$ a structure of graded $\cstarr$ by conjugation with $\epsilon$ for which the map $\phi: A \to \endo ({\hilb})$ is a graded $*-$morphism
and $F$ is an odd operator of $\endo( \hilb)$ satisfying for any $a$ in $A$ the three conditions:
\begin{align}\label{fredmodcond}
  (F-F^*)\phi(a)  \in \compact && (F^2 - \id)\phi(a) \in \compact && [F, \phi(a)] \in \compact
\end{align}
Coboundaries are those triples such that the compact operators of the former equation all vanish.
We obtain a contravariant functor $\KO^0$ from $\cstarrs$ to abelian groups.
For example if $M$ is a smooth compact manifold and $P$ is an pseudo-differential elliptic operator acting between smooth sections of two real vector bundles.
A parametrix $Q$ of it defines an operator $D = \left( \begin{smallmatrix}
  0 & P \\ Q & 0
\end{smallmatrix} \right)$ a class in $\KO(\fun(M))$ by
$[\summable^2(M, E) + \summable^2(M, F), \phi, D \sqrt{1 + D^2}]$ where $\phi$ is multiplication on sections.
This class does not depend on the choice of the parametrix.
As for $\KO$ cohomology, one define higher $\KO$ homology groups by declaring $\KO^n(A) = \KO(S^n A)$

For example the $\KO$ homology of the scalar algebra $\R$ is given by classes $[ \hilb, \phi, F]$ where the representation $\phi$ of $\R$ is given by a projection $p = \phi(1)$ and where $F = \left( \begin{smallmatrix}
  0 & P \\ Q & 0
\end{smallmatrix} \right)$ is given by a Fredholm operator $P: p \hilb^0 \to p  \hilb^1$.
The index $\dim \ker P - \dim \ker P^*$ of this operator gives a map $\KO_0(\R) \to \Z$ that is in fact an isomorphism.
Higher $\KO$ homology groups of $\R$ are the (reduced) $\KO$ homology of algebras of spheres.
Dirac operators then define elements that are zero in $\KO^1, \KO^2, \KO^3$ and $\KO^5$, a free generator of $\KO^4$ and $\KO^8$ and a 2 torsion generator of $\KO^6$ and $\KO^7$.
The degree 8 element of $\KO^8$ induces isomorphisms $\KO^n(A) \to \KO^{n+8}(A)$ and as before we obtain a family of contravariant functors that are homotopy invariant, stable and take short exact sequences to long exact sequences of order 24.

The map given at the level of cycles by
$\left \langle  [ \hilb, \phi, F]  ,  [p]  \right \rangle  = [ \hilb \gotimes  \compact, (\phi \gotimes \id)(p), F \gotimes \id ]$ defines the bilinear index pairing $\KO^0(A) \times \KO_0(A) \to \KO_0(\R)$ with value in $\Z$ after taking the Fredholm index.
More generally the same formula defines a pairing $\KO^n(A \gotimes B) \times \KO_m(A) \to \KO^{n-m}(B)$.
We will focus on the pairing $\KO^{n}(A) \times \KO_{n+i}(A) \to \KO^{-i} = \ZZ$ for $i = 1$ or 2.
Clifford algebras play a similar role than the suspension algebras in such a theory.
This will be explained further in the framework of $\KKO$ theory that we recall now.

The symmetry between $\KO$ cohomology and $\KO$ homology and the pairing are better seen through the bivariant $\KKO$ theory.
Let two graded $\cstarrs$ $A$ and $B$.
We consider following~\cite{k81} and~\cite{cs84} \emph{Kasparov bimodule} defined as triples $(E, \epsilon, \phi, F)$, denoted simply $(E, \phi, F)$ if the context is clear, where $(E, \epsilon)$ is a graded real $B-$Hilbert module on the right, $\phi$ a graded $*-$morphism from $A$ to $\endo_B(E)$ and $F$ is an odd operator of $\endo_B(E)$ verifying:
\begin{align}\label{kasmodcond}
  (F-F^*)\phi(A)  \subset \compact(E) && (F^2 - \id)\phi(A) \subset \compact(E) && [F, \phi(A)] \subset \compact(E)
\end{align}
We say that such a Kasparov triple is \emph{degenerate} if the conditions in~\cref{kasmodcond} can be strengthened to the vanishing of the involved compact operators.

Modded out by homotopic triples it becomes an abelian group for the direct sum denoted by $\KKO(A, B)$.
The neutral element is given by the class of any degenerate Kasparov bimodule.
Classes of bimodules $(E, \phi, F)$ will be denoted $[E, \phi, F]$.
Higher $\KKO$ groups are then defined as $\KKO_n(A, B) = \KKO(A, S^n B)$.

Such functors are covariant in the second variable and contravariant in the first variable.
They extend the two monovariant functors: $\KKO_n(A, \R) \simeq \KO_{-n}(A)$ and $\KKO_n(\R, A) \simeq \KO^n(A)$.
The first isomorphism is straightforward, the second is given by an argument due to Karoubi~\cite{k70} in the commutative setting and by Roe~\cite{r04} in the non-commutative one.
The argument goes in two parts.
The first identifies from the definition $\KKO(\R, A)$ with $\kok_1(\calkin \otimes A)$ for $\calkin$ the Calkin algebra.
Secondly, this $\kok$ group is isomorphic to $\kok(A)$ via the boundary map for the exact sequence $\compact  \to \bounded\to \calkin $.

The product defined by Kasparov for $D$ separable and $\sigma-$unital takes the general form
\[\KKO_n(A_1, B_1 \gotimes D) \times \KKO_m(A_2\gotimes D, B_2) \to \KKO_{n+m}(A_1 \gotimes A_2, B_1 \gotimes B_2)\]
This pairing is associative and extends the index pairing between $\KO$ theory and $\KO$ homology.
For any $\cstarr$  $A$ there is a neutral element $1_A$ in $\KKO(A, A)$ given by the class of $(A, 0)$.
Functoriality in both variable for a $*-$morphism $\phi : A \to B$ can be expressed as the product with $\bm \phi = [B, \phi, 0] $ in $ \KKO(A, B)$.

If $A$ and $A'$ are two trivially graded Morita invariant $\cstarrs$ through bimodules $E$ and $F$ and isomorphisms $E\gotimes_A F \simeq A'$ $F\gotimes_{A'} E \simeq A$ pairing with $[E, 0]$ and $[F,0]$ induce inverse isomorphisms $\KKO(A, B) \simeq \KKO(A', B)$ and $\KKO(B, A) \simeq \KKO(B, A')$.
This is still true for graded $\cstarrs$.
Note however that for an algebra of compact operators on an eventually finite dimensional Hilbert space $\hilb$ a grading is always given by some $x \mapsto u x u^*$ with $u$ bounded on the Hilbert space and with square a multiple of the identity $\pm 1$.
This algebra is Morita equivalent to $\R$ if and only if $u^2= 1$ and we can take for the Morita equivalence $\hilb$ and its dual given with the grading $u$.
In this case, the grading is said to be \emph{even}, in the other case it is said to be \emph{odd}.
For example the grading of $\clif[2,0] \simeq \M_2(\R)$ is odd given by conjugation with
$\left(\begin{smallmatrix}
  0 & -1 \\ 1 & 0
\end{smallmatrix}\right)$.

Kasparov defined two $\KKO$ elements by:
  \[\bm{\alpha_{p,q}} =  \left[\l^2 (\R^{p+q}) \gotimes \lambdarpq , \dirac/ \sqrt{1 +\dirac^2} \right] \in \KK(\func(\Rpq), \clif[q,p] ) \]
  \[\bm{\beta_{p,q}} =  \left[\fun_0 (\R^{p+q}) \gotimes \lambdarpq , \dualdirac / \sqrt{1 + ||x||^2} \right] \in \KK(\clif[q,p], \func(\Rpq))\]

  Where a generator of $\clifpq$ acts on $\lambdarpq$ as:
  \begin{equation}
    \label{clifaction}
     \omega \mapsto x \wedge \omega \pm \iota_x(\omega)
  \end{equation}
  With this action $\lambdarpq$ is a free $\clif[q,p]$ module of rank 1. The action can then be identified with Clifford multiplication on itself, from the right or from the left.
  
  Kasparov obtained the Bott periodicity as a \emph{$\KKO$ equivalence}:

\begin{theo}\label{bottkas}[Bott-Kasparov]
$\bm \alpha_{p,q} \cdot \bm \beta_{p,q} = \bm 1  \in \KKO(\func(\Rpq), \func(\Rpq))$ and
$\bm \beta_{p,q} \cdot \bm \alpha_{p,q} = \bm 1  \in \KKO(\clifpq, \clifpq)$
\end{theo}

Kasparov product with these elements gives isomorphisms
$\KK_n(A, B) \simeq \KK(A, \clif[0,n] \gotimes B)$.
Now, the action of~\cref{clifaction} gives an isomorphism between $\clifpq \gotimes \clif[q,p] $ and $ \endo(\lambdarpq)$ given its even grading, and then Morita equivalent to $\R$.
We have:
\[\KK(\clifpq \gotimes A, \clif[r,s]\gotimes B) \simeq \KK_{p - q + s - r}(A, B)\hspace{1cm} \text{ for any }p, q, r,s  \text{ in } \N\]

In fact the isomorphism $\KK(\clifpq\gotimes A, B) \to \KK(\clifpq \gotimes \clif[q,p] \gotimes A , \clif[q,p]\gotimes B) \to \KKO(A,\clif[q,p]\gotimes B)$
given by the external product with the unit of $\KKO(\clif[q,p], \clif[q,p])$followed by Morita invariance of $\KK$.
Explicitly:
\begin{lemma}\label{cliflr}
  For $[E, \epsilon, c \gotimes \phi, F] $ in $\KKO(\clifpq\gotimes A, B)$ the corresponding element of $\KKO(A, \clif[q,p]\gotimes B)$ is given by the class of
  $[E, \epsilon ,\phi, F']$ where $E$ is considered as an Hilbert $\clif[q,p] \gotimes B$ module with right action of a generator $c_i^o$ of $\clif[q,p]$ given as $c_i \epsilon$.
  The $\clif[q,p]\gotimes B$ bilinear product is given by $\langle e, f\rangle = \frac{1}{2^n}\sum_I  c_I^* \gotimes ec_I, f)$ and we have the formula for the operator $F '  =  \frac{1}{2^{n}}\sum_{I } {(-1)}^{I} c_I^{o*}  F c_I^{o}$.
\end{lemma}

\begin{remark}
  If $F$ was commuting with the initial $\clifpq$ left action, then $F' = F$.
 In fact $\frac{1}{2}(F - c_i F c_1^*)$ is a compact perturbation of $F$ anti-commuting with $c_1$.
 Changing $F$ for $\frac{1}{2}(F - c_i F c_1^*)$ give the same $\KKO$ class and making the successive changes for $i$ ranging from $1 $ to $p+q$ gives exactly $F'$ as in the previous lemma
\end{remark}

\begin{remark}
  The same procedure takes a right $\clif[p,q]$ action to a $\clif[q,p]$ action on the left.
  The Fredholm operator of stay unchanged as it already commute with the Clifford right action.
  It may seems at a first glance that information on the $\clifpq$ valued pairing is lost in the process.
  In fact, as shown in the lemma, one can recover it from the Clifford action itself.
  Actually, modulo multiplication by a positive scalar this pairing is the $\clifpq\gotimes B$ inner product that extends the $B$ valued one in a $B-$Hilbert module with $\clifpq$ right action.
  In what follows, we can then forgot about this $\clifpq$ part of the bilinear pairing.
  A Clifford algebra $\clifpq$ appearing as an argument in $\KKO_{\pm (p -q)}(A, B)$ is interpreted as a additional multi-grading on the Kasparov bimodule on $A-B$.
\end{remark}

If we denote by $S^{p,q}A$ the algebra $\func(\Rpq)\otimes A$, using once again the Bott periodicity, we obtain the isomorphisms $\KK(S^{p,q}A, S^{r,s}B ) \simeq \KK_{r -s -q + p}(A, B)$ for any positive integers $p,q,r,s$

The algebra $\clif[8,0]$ is isomorphic with the matrix algebra $\M_{16}(\R)$ with even grading.
This gives the periodicity of Clifford algebras and shows the theorem of periodicity in $\KKO$ theory:

\begin{theo}
  For any couple of graded $\cstarrs$ $A$ and $B$, the product with $\bm \alpha_{8,0}$ and $\bm \beta_{8,0}$ induces an isomorphism 
$\KKO_{n + 8}(A, B) \simeq \KKO_n(A, B) $
\end{theo}

\section{Real versus complex bivariant \texorpdfstring{$\K$}{K} theories}
\subsection{Realification and complexification in bivariant \texorpdfstring{$\K$}{K} theory}
We defined real and complex bivariant theory for real and complex $\cstars$.
But, a complex $\cstarr$ can be considered as a $\cstarr$.
In the other way, if $A$ is a real $\cstarr$ the algebra $A_\C = A \otimes \C$ is a complex $\cstar$.
For commutative $\cstarrs$, in this process, the real structure given by the proper involution on the spectrum of the algebra is lost : $\func(X, \tau)_\C = \func_\C(X)$ 
We study in this section the two functors that play the same role for Kasparov triples, going from complex theory to the real one and vice-versa.
This extends realification and complexification of vector bundles on spaces.

In~\cite{a66}, Atiyah shows that for any space $X$ there is a canonical isomorphism $\KO(X_\C) \simeq \KU(X)$ where $X_\C$ is the real space given by $X \cup X$ and involution given by swapping.
We begin by a non-commutative bivariant extension of Atiyah's argument showing that we can recover the functor $\KKU$ with the functor $\KKO$:

\begin{lemma}\label{kkr} For $A$ and $B$ two $\cstarrs$ we have a canonical identification between the groups $\KKO(A_\C, B)$, $\KKO(A, B_\C)$ and $\KKU(A_\C, B_\C)$, compatible with the Kasparov product.
\end{lemma}

\begin{proof}
  We assume for simplicity that $A$ and $B$ are unital.
  
  Starting with a real $A- B_\C$ bimodule, the action of $i$ on the right gives an action of $i$ on the left and a representation of $A_\C$.
  
  Starting with a real $A_\C-B$ bimodule, the same process gives an action of $B_\C$ on the Hilbert $B-$module.
  We make it a Hilbert $B_\C-$module with $\left\langle e, f \right \rangle =\frac{1}{2} ((e,f) + i (e, if))$.
  In the case where the operator $F$ does not commutes with this action we change it to $\frac{F - i F i }{2}$.

  If we have now a complex $A_\C- B_\C$ Kasparov bimodule, one can forget the structure of $i$ on the left to obtain a real bimodule on $A-B_\C$.
  To obtain a real bimodule on $A_\C- B$ we consider the Hilbert $B-$module obtained by forgetting the action of the imaginary part and by the formula for the scalar product $\left\langle e, f \right \rangle =\re \left( e, f \right)$.
\end{proof}

\begin{remark}
  We must take care that $\KKO(A_\C, B_\C)$ is not isomorphic to $\KKU(A_\C, B_\C)$. Actually the lemma shows that $\KKO(A_\C, B_\C) \simeq \KKU(A_\C, B_\C) + \KKU(A_\C, B_\C)$ as $\C \otimes \C = \C \oplus \C$.
  On a real $A_\C- B_\C$ Kasparov bimodule, the action of the imaginary coming from the left may not be equal to the one coming from the right.
  We can just say that they commute.
  We can then split our space in two spaces, one where those two actions coincide, one where they are opposite.
  This decomposition gives the splitting of the $\KK$ group above.
\end{remark}

  Now we come to the functors of realification and complexification in bivariant $\K$ theory defined at the level of Kasparov triples:

  \begin{defi}
    Let $A$ and $B$ be two real $\cstars$ and a real bimodule $(E, \phi, F)$ on $A-B$ we define its complexification  ${(E, \phi, F)}_\C $ as the bimodule on $A_\C-B_\C$ given by $(E_\C, \phi_\C, F_\C)$.

    In the opposite direction, given a complex $A_\C-B_\C$ bimodule  $(E, \phi, F)$ we consider $E_\R$ which is the same normed real vector space as $E$ where we forgot the complex action and we take the same $B$ valued scalar product as in the previous lemma.
    The action of $A_\C$ restricts to an action $\phi_\R$ of $A$ and we define the realification ${(E, \phi, F)}_\R$ of $(E, \phi, F)$ as the real bimodule on $A-B$ given by $(E_\R, \phi_\R, F)$.
  \end{defi}

  In $\KK$ theory now, we can link these  of complexification and realification with the previous lemma and with the functoriality of the $\KK$ groups with respect to the two morphisms $i: \R \to \C$ and $r: \C \to \M_2(\R)$:
 \begin{lemma}\label{realcomp}
 For $A$ and $B$ two $\cstarrs$, we have the two functors of realification and complexification between $\KKO$ and $\KKU$ given by any path in the corresponding commutative diagram:
\begin{center}{\begin{tikzcd}[column sep=tiny]
	& \KKO(A_\C, B) \arrow[dr]	&			&	& \KKO(A_\C, B) \arrow[dr, leftarrow]	&	\\
\KKO(A,B) \arrow[ur, "\bm r \cdot  "] \arrow[dr, "\bm i \cdot  "]&						& \KKU(A_\C, B_\C) & \KKO(A,B) \arrow[ur, leftarrow,  "\bm i \cdot_\C"] \arrow[dr, leftarrow, "\cdot_\C \bm r"]&						& \KKU(A_\C, B_\C) \\
	& \KKO(A, B_\C) \arrow[ur]	&	&	& \KKO(A, B_\C) \arrow[ur, leftarrow]	&
			\end{tikzcd}}
		\end{center}
Moreover, these functors are compatible with the Kasparov product.
\end{lemma}

\begin{remark}
This lemma shows that both $ \KKO(\C, \R)$ and  $\KKO(\R, \C)$ are isomorphic to $\Z$, respectively generated by $ \bm r $ and $  \bm i $.
\end{remark}

\subsection{The exact sequence}
The two transformations of complexification and realification are not inverse one from the other.
Starting from a real K theory class, complexifying and then forgetting the complex structure will give twice the initial class.
We give a more precise relation between these two functors in the form of a long exact sequence, implying a particular 2-torsion element, the generator $\eta$ of $\KO_1$.

We consider the real space $\R^{0,1}$ and the inclusion of the origin in this space.
This gives an exact sequence of real algebras.
	\[0 \to \func(\R^{0,1} \setminus \{0\}) \to \func(\R^{0,1} ) \totxt{ev_0} \R \to 0 \]
	 Now $\func(\R^{0,1} \setminus \{0\})$ is isomorphic to $\func_\C(\R)$ by restricting to the positive reals.
 For any couple $A, B$ of $\cstarrs$ we can now derive the long exact sequence:
\[\cdots \to \KKO_1(A, B) \to \KKO(A, S^1_\C B) \to \KKO(A, S^{0,1}B) \to \KKO(A, B) \to \cdots \]
And by Bott periodicity~\cref{bottkas} we have:
\begin{equation}
  \label{kw}
\cdots \to \KKO_1(A, B) \to \KKU_1(A_\C, B_\C) \to \KKO_{-1}(A, B) \to \KKO_0(A, B) \to \cdots  
\end{equation}

We now identify arrows in the diagram.
Those are induced by Kasparov products with elements in some very simple $\KK$ groups.

\begin{prop}
$\bm{ev}_0 \in\KKO_1(\R, \R)$ is given by the generator $\eta$ of $\KO_1$.
This also identifies with $\bm j$ where $j$ is the inclusion $ \R \incl \clif[0,1]$ and with $\bm k$ where $k:\clif[1,0] \incl \M_2(\R)$ sends the generator $e$ to
$(\begin{smallmatrix}
 0 & 1 \\ 1 & 0  
\end{smallmatrix})$
.
\end{prop}
\begin{proof}
  The exact sequence for $A = \R$ gives:
  \[\KU_2 \to \KO_0 \totxt{\bm{ev}_0} \KO_1 \to \KU_1 \]
  We see that $\KU_1= 0$, $\bm{ev}_0: \KO_0 \to \KO_1$ is onto and $\bm{ev}_0 \neq 0$ in $\KO_1 = \ZZ$, $\bm{ev}_0$ identifies with the generator $\eta$ of $KO_1$ as an element of $\KKO_1(\R, \R)$.

  Consider now $\bm{ev}_0$ in $\KK(\func (\interoo{-1,1}^{0,1}), \R)$.
  The Kasparov product with the Bott element $\beta_{0,1} \in \KK(\R, \func (\interoo{-1,1}^{0,1})\otimes \clif[1,0])$ over $\func(\interoo{-1,1}^{0,1})$ is given by functoriality.
  Evaluating $\beta_{0,1}$ in 0 we find:

  \[\bm\beta_{0,1}\cdot_{\func(\interoo{-1,1}^{0,1})} \bm{ev_0} = ev_{0*}(\bm\beta_{0,1})= \left[\clif[0,1], 0 \right] = \bm j \in \KKO(\R, \clif[0,1])\]
  Now taking the Clifford algebra to the left in $\KKO$ as explained in~\cref{cliflr} we obtain on $\R^2$ the action of $\clif[1,0]$ given by 
$(\begin{smallmatrix}
 0 & 1 \\ 1 & 0  
\end{smallmatrix})$.
\end{proof}

 In the same way, $\bm \delta$ can be seen as an element of $\KKO(\R, \C)$.
 This group is isomorphic to $\Z$, generated by $\bm i$ for $i$ the inclusion $\R \incl \C$.
 We extract the exact sequence
\[\KO_{-1} \to \KO_0 \totxt \delta \KU_0 \to \KO_{-2} \]
Identifying $\KO_0(\R)$ and $\KU_0(\C)$ with $\Z$ we see that $\delta $ is an automorphism of $\Z$ and is then equal to $\pm id$. Under the same identification $i$ induce the identity on $\Z$. We conclude that $\bm \delta = \pm \bm i$.
This amounts to say that the generator of the reduced $\KO$ theory of the sphere $S^2$ is the realification of the Bott generator generating the reduced $\KU$ of $S^2$.

We give another proof that will shed some light on this part of the exact sequence.
The morphism $\KO_1(\R) \to \KU_1(\C)$ is the boundary map of the long exact sequence associated to the sequence of spaces above.
It can be constructed using the Puppe sequence.
This gives the long exact sequence associated to any CW pair by the following arguments.
For $F$ a closed subset of $X$ with complementary $U$, we construct the cone of $F$ in $X$ as $\cone{F/X} = \interoo{0,1}\times F \cup {0}\times X $.
This cone is homotopy equivalent to $U\sim \cone{F/X}$ and $X$ can be seen as a closed subset of it. The complementary of $X$ in the cone is the suspension $SF$ of $F$: $ SF= \interoo{0,1}\times F$.
The boundary map in any cohomology theory is then induced by the inclusion of the open set $SF$ in $ \cone{F/X} \simeq U$.
This argument computes the sign, which is not of real interest:

\begin{prop}
	The arrow $\KKO_n(A, B) \to \KKU_n(A_\C, B_\C)$ in the long exact sequence~\cref{kw} is given by $\bm \delta = \bm i$
\end{prop}

\begin{proof}
The exact sequence can be derived from the Puppe construction that we will explicit at the level of algebras.
The algebra $\fun(S^{0,1})$ appears in fact as the cone $\cone{i}$ for $i$ the inclusion $\R \to \C$.
More precisely the morphism $ev_0$ gives the cone $C_{ev_0}$ by the following:

  \[C_{ev_0} = \{(f , g) \in \func(\interoo{-1,1})_\C \oplus \func(\interco{0,1}) | f(-x) = \overbar{f(x)} , f(0) = g(1)\}\]
The injection $I: S\R \to C_{ev_0}$ is given by $f \mapsto (0, f)$ and there is an equivalence of homotopy between $C_{ev_0}$ and $\func(\interoo{0,1})\otimes \C$ given by
\begin{center}
  $\begin{array}{clllcll}
\Phi: C_{ev_0}								& \to 		&	\func(\interoo{0,1})_\C 		&	&
\Psi: \func(\interoo{0,1})_\C & \to & 	C_{ev_0}\\
(f, g)										& \mapsto 	&	x\mapsto 	\begin{cases} g(2x)	& \text{pour }x \leq\frac{1}{2} \\
                                                 f(2x-1)	& \text{pour }x \geq \frac{1}{2}
                                        \end{cases} 				& &
f 											& \mapsto 	& 	x\mapsto	\begin{cases} f(x)	& \text{pour }x \geq 1 \\
                                                  \overbar{f(-x)}	& \text{pour }x \leq 1
                                        \end{cases}
  \end{array}$
\end{center}

We represent this situation with the following picture of the involved real spaces, where the real structure is given by the arrow when not trivial:

\begin{center}
\begin{tikzpicture}
  \draw  (1,0.2) node  {$\to$};
  \draw[)-(] (2,0) -- (2,1);
\draw[)-(] (0,0) -- (-2,0);
\draw [-(](-1,0) -- (-1,1);
\draw [<->](-0.2,-0.2) -- (-1.8, -0.2);
\draw  (-3,0.2) node  {$\simeq$};
\draw [)-(](-4,0) -- (-4.9,0);
\draw [)-(](-5.1,0) -- (-6,0);
\draw [<->](-4.2,-0.2) -- (-5.8, -0.2);
\end{tikzpicture}
\end{center}

The map $\delta: \KO_1 \to \KU_1$ is then given by composition $\Phi \circ I: S\R \to C_{ev_0} \to \func(\interoo{0,1})_\C $ with $\Phi \circ I (f)(x) = 0$ if $ 0 \leq x \leq \frac{1}{2}$ and $\Phi \circ I (f)(x) = f(2x-1)$
if $ \frac{1}{2} \leq x \leq 1$.
This map is homotopic to $i \otimes \id: \func(\interoo{0,1}) \to  \func(\interoo{0,1})_\C $ given by inclusion $i: \R \incl \C$.
And, the boundary map is $\delta = \bm i \in KKO(\R, \C)$, the complexification morphism.
\end{proof}

The last morphism is given by an element of $\KKO_{-2}(\C, \R)$.
Using lemma~\cref{realcomp} , this group is isomorphic to $\Z$, generated by $\bm r \bm \beta^{-1}$ for $r: \C \to \M_2(\R)$ the realification morphism $r(x+ i y) = (\begin{smallmatrix}x & -y\\y & x \end{smallmatrix})$.

The exact sequence~\cref{kw} for $A= \R$ identifies the boundary map $KU_2 \to \KO_0$ as $\pm2\id$.

\begin{center}{\begin{tikzcd}[row sep = tiny, column sep = tiny]
	\cdots \arrow[r]			& 	\KU_2 \arrow[r]  &	\KO_{0}\arrow[r] 	&	\KO_1 \arrow[r] 	&	\KU_1 \arrow[r]
&	\cdots \\
	\cdots \arrow[r]			& 	\Z \arrow[r]	& 	\Z \arrow[r]	& 	\ZZ \arrow[r]	& 	0\arrow[r]
&\cdots \\
			\end{tikzcd}}
		\end{center}

Using the same identification, realification composed with the Bott map gives $2\id$.
The two maps coincide modulo multiplication by a sign.

In fact we can compute the sign:

 \begin{prop}
		The arrow $\KKU_n(A_\C, B_\C) \to \KKO_{n-2}(A, B)$ in the long exact sequence~\cref{kw} is given by $\bm \phi=  \bm r \circ \bm \beta^{-1} \in \KK_{-2}(\C, \R)$
\end{prop}

\begin{proof}
  Let the real structure on $\M_2(\C)$ given by:
  \[\overbar{\left(\begin{matrix}\hphantom{\overbar d}a & b\\c & d\end{matrix}\right)} =  \left(\begin{matrix}\overbar d & \overbar c\\ \overbar b & \overbar a\end{matrix}\right)\]
  With this real structure we can identify $\M_2(\C)_\R$ with $\M_2(\R)$ and $S^{0, 1}\M_2(\R) $ with $ \{f: \interoo{-1,1} \to \M_2(\C), f(-x) = \overbar{f(x)}\}$.
  Define for $f \in \func_\C(\interoo{0,1})$:
  \[\phi_tf(x) = \left(\begin{matrix}f((1-\frac{t}{2}) x + \frac{t}{2}) & 0\\0 & \overbar{f((1-\frac{t}{2}) x + \frac{t}{2})}\end{matrix}\right)\]
  Where $f$ is extended by zero out of its domain.
  This makes explicit an homotopy between the two following morphisms from $\func_\C(\interoo{0, 1}) $ to $ \M_2\left(\func(\interoo{-1, 1}^{0, 1})\right) $:
  \begin{align*}
	(\id \otimes \phi) \circ p (f)(x)	=& \left(\begin{matrix}f(x)& 0\\0 & \overbar{f(-x)}\end{matrix} \right)
  						&&	\phi'(f)(x) =&	\left(\begin{matrix}f(2x-1)& 0\\0 & \overbar{f(1-2x)}\end{matrix} \right)
  \end{align*}
   Where $\phi $ is the map  :$ S_\C \otimes \to S^{0, 1} $ from the exact sequence is tensored with $id : \M_2 \to \M_2$ and is composed with a rank one projector $p: \C \to \M_2(\C)$
 making explicit the Morita invariance isomorphism. 
Now the isomorphism $\psi: S^{0,1}_ \C \to S^1_\C $ given by $f\otimes z \mapsto zf$ composed with  $\phi '$ gives a map $S^{0,1} \C\to S^{0, 1}\M_2(\R) $ that coincides with $\id \otimes i$ with $i:\C \incl \M_2(\R)$.
   In fact the element $\bm \psi $ of $ \KKU( S^{0,1}_\C , S^{1,0}_\C)$ is exactly the Bott generator $\bm \beta$ of $\KU_2$.
\end{proof}
\begin{remark}
Taking the exact sequence for the algebra $A = \H$ and using the isomorphism  $\C \otimes \H \simeq \M_2(\C)$, we obtain:

\begin{center} \label{kwh}{\begin{tikzcd}[row sep = tiny, column sep = tiny]
	\cdots \arrow[r]			& 	\KO_2(\H) \arrow[r]  &	\KU_2(\C)\arrow[r] 	&	\KO_0(\H) \arrow[r] 	&	\KO_1(\H) \arrow[r]
&	\cdots \\
	\cdots \arrow[r]			& 	0 \arrow[r]	& 	\Z \arrow[r]	& 	\Z \arrow[r]	& 	0\arrow[r]
&\cdots \\
			\end{tikzcd}}
		\end{center}

In fact here $r \otimes \H: M_2(\C) \to \M_2(\H)$ is given by an inclusion $i': \C \incl \H$. We could have given a similar proof that the one presented above using the real structure on $\M_2(\C)$ given by
  $\overbar{\left(\begin{smallmatrix}\vphantom{\overbar d}a & b\\c & \vphantom{\overbar d}d\end{smallmatrix}\right)} =  \left(\begin{smallmatrix}\phantom{\shm}\overbar d & \shm \overbar c\\ \shm\overbar b & \phantom{\shm}\overbar a\end{smallmatrix}\right)$ for which the real locus is isomorphic to $\H$.
\end{remark}

Identification of all arrows is done.
The exact sequence commutes with the Kasparov product by functoriality of the Kasparov product.
We sum up:

  \begin{prop}\label{kwprop}
  Let $A, B$ be two graded $\cstarrs$.
  The following exact sequence holds:

  \begin{tikzcd}[column sep = small]
   \cdots \arrow[r] &	\KKO_{n}(A, B) \arrow[r,"c"]			& 	\KKU_{n}( A_\C, B_\C) \arrow[r, "r \beta^{-1}"]
  &	\KKO_{n-2}(A, B) \arrow[r,"\eta"] 	&	\KKO_{n+1}(A, B) \arrow[r]
  &	\cdots
  			\end{tikzcd}\\
  Where arrows are described as:
  \begin{itemize}
    \item $\eta$ the generator of $\KO_1$, equal to $\bm j_* \in \KKO(\R, \clif[0,1])$ induced by the inclusion $j: \R \to \clif[0,1]$
    \item $c$ the complexification morphism
    \item $r\beta^{-1}$ Bott periodicity composed with the realification morphism
  \end{itemize}
  Let $D$ be an other graded $\cstarr$ and let $\bm \alpha $ in $ \KKO_k(D, B)$.
  The morphism $\KO_*(A) \to KO_{*+k}(B)$ induced by the Kasparov product with $\bm \alpha$ commutes with the exact sequence:

  \begin{tikzcd}[column sep = tiny]
     \cdots \arrow[r]  & \KKO_n(A,D) \arrow[r]\arrow[d, "\bm{\alpha}"]			& 	\KKU_n(A_\C, D_\C) \arrow[r] \arrow[d, "\bm{\alpha}_\C"]
  &	\KKO_{n-2}(A, D) \arrow[r] \arrow[d, "\bm{\alpha}"] 	&	\KKO_{n-1}(A,D) \arrow[r] \arrow[d, "\bm{\alpha}"]
  &	\cdots \\
   \cdots \arrow[r]  &	\KKO_{n+k}(A, B) \arrow[r]			& 	\KKU_{n+k}( A_\C, B_\C) \arrow[r]
  &	\KKO_{n+k-2}(A, B) \arrow[r] 	&	\KKO_{n+k-1}(A,B) \arrow[r]
  &	\cdots
        \end{tikzcd}
  \end{prop}

\begin{remark}
  The exact sequence was derived from the specialization of it to the algebra of scalars, which amount to the relations between generators of $\KU_*$ and $\KO_*$:
                                                                                                                                                      \begin{align*}
\eta_\C &= 0 & \gamma_\C &= 2 \beta^2& \delta_\C &= \beta^4 \\ \beta_\R &= \eta^2& \beta^4_\R&= \gamma & \beta^4_\R &= 2 \delta                                                                                                           \end{align*}

\end{remark}

\begin{cor}
   We have 4-periodicity in $\KKO$ after inverting 2: $ \KKO_i(A,B)[\frac{1}{2}] \simeq KKO_{i+4}(A,B)[\frac{1}{2}] $.
  This follows from the equality $\gamma^2 = 2 \delta$.
  If we inverse 2 in the exact sequence, $\eta$ vanishes and we obtain the following exact sequence:
  $0 \to\KKO_{i+2}(A,B)[\frac{1}{2}] \to\KKU_i(A_\C, B_\C)[\frac{1}{2}] \to\KKO_i(A,B)[\frac{1}{2}] \to 0$.
  In fact realification and complexification induce an isomorphism
  \[\textstyle\KKO_i(A, B)[\frac{1}{2}]\oplus\KKO_{i+2}(A, B)[\frac{1}{2}] \simeq \KKU_i(A_\C, B_\C)[\frac{1}{2}] \]
\end{cor}

  It also shows a generalization of a monovariant result of Karoubi~\cite{k01} already given by Schick~\cite{s04}: $\KKO_i(A, B) = 0$ for every $i$ if and only if $\KKU_i(A_\C, B_\C) = 0$ for every $i$.
  The vanishing of $\KKO$ directly implies the vanishing of $\KKU$ by the exact sequence.
  Now if $\KKU$ vanishes, one has to use an additional argument, the nilpotency of $\eta$.
  If $x \in \KKO_i(A, B)$ then the complexification of $x$ is zero.
  We can write $x = \eta y$ but once again because of the vanishing of $\KKU$ the complexification of $y$ is zero.
  We can write $x = \eta^3 z$ and $x$ is then 0.
In fact, a similar proof gives more:
\begin{cor}
  Let $\bm \alpha \in \KKO(A, B)$ then $\bm \alpha$ is invertible if and only if $\bm \alpha_\C$ is.
\end{cor}
The same statement being false for monomorphisms and epimorphisms as can be seen from the elements induced by a degree 2 map of $S^1$ onto itself and by realification $\R \incl \C$.

\section{Anti-linear symmetries} 

\subsection{Higher \texorpdfstring{$\KKO$}{KKO} theory via real structures}\label{secrealkko}

Study of Clifford algebra representations leads to the notion of even and odd complex Kasparov triples.
We will see in this section how to interpret this in the real setting.

Even complex Clifford algebras ${\cclif[2k]}$ are isomorphic to $\M_{2^k}(\C)$ with grading given by a diagonal matrix with the same number of coefficients 1 and $-1$ on the diagonal.
This graded algebra has two isomorphism classes of irreducible graded modules.
The two classes are swapped after inversion of the grading.
We can then denote $R_{2k}$ and $R_{2k}^o$ for two graded modules representing those two classes.
Both are isomorphic as ungraded $\cclif[2k]$ modules to the regular representation on $\C^{2^k}$.
Odd Clifford algebras ${\cclif[2k+1]}$ are isomorphic to $\M_{2^k}(\C) + \M_{2^k}(\C)$ with even elements of the shape $(a ,a)$ for $a\in \M_{2^k}(\C)$ and odd ones given by some $(a ,-a)$.
For this graded algebra there is only one isomorphism class of irreducible graded module given by  $R_{2k+1}$ the standard representation of $\M_{2^k}(\C) + \M_{2^k}(\C)$ on $\C^{2k} + \C^{2k}$ with grading interchanging the two components.

Take now  two ungraded complex $\cstars$ $A$ and $B$.
Any $B \otimes \cclif-$ graded Hilbert module $E$ can be written $\tilde E_0 \gotimes R + \tilde E_1 \otimes R^o = \tilde E \gotimes R$ where $\tilde E = \tilde E_0 + \tilde E_1$ is a graded Hilbert $B-$module. 
Any $B\otimes \cclif-$ endomorphism of $E$ is given by a unique $M \otimes \id $ or $M \gotimes \id$ for $M$ in $\endo_B(\tilde E)$.
It is the Schur Lemma for complex Clifford algebras.
Now for a Kasparov triple of the form $(E, \phi, F)$ on $A - \otimes B \otimes {\clif[n]}_\C$ the action of $A$ can be written $\phi = \tilde \phi \otimes \id$.
and  $F$  as  $ \tilde F\gotimes \id $ where $\tilde F$ anti-commutes with the grading on $\tilde E$.
One reduces to a complex $A -  B$ Kasparov module: $(\tilde E, \tilde \phi, \tilde F)$.

In the other case the Hilbert module over $B \otimes \cclif$ can be written $\tilde E \otimes R$ for $\tilde E$ an Hilbert $B-$module.
The grading is then only supported on $R$. 
the only irreducible graded $\cclif$ module.
recall that we have an isomorphism $\omega : R \to R^o$ assumed to square to $\id$ after Schur lemma and a rescaling.
We can now write$\phi$ as $\tilde \phi \otimes 1$ and $F$ as $ \tilde F \otimes \omega$.
This is the motivation of the definition of odd Kasparov triples :

\begin{defi}
  For $A$ and $B$ two complex ungraded $\cstars$, we call $(E, \phi, F)$ an odd Kasparov bimodule on $A-B$ if
  $E$ is an Hilbert $B-$module, if $\phi: A \to \endo_B(E)$ and if $F \in \endo_B(E)$ verifies :
  \begin{align*}
    (F - F^*) \phi(a) \in \K(E) &&(F^2 - 1) \phi(a) \in \K(E) && [F, \phi(a)] \in \K(E) && \forall a \in A  
  \end{align*}

\end{defi}

The Hilbert $B-$ module $E$ is ungraded and commutators $[F, \phi(a)]$ are defined as  $F \phi(a) - \phi(a) F$.
Our previous discussion shows the following formal Bott periodicity:

\begin{prop}
  Let $A$ and $B$ be two ungraded complex $\cstars$.

  There is an equivalence between Kasparov bimodules on $A-B$ and Kasparov bimodules on $A- B \otimes \cclif[2k] $  given by assigning the Kasparov bimodule  $( E, \phi, F)$ on $A-B$ to $(E \gotimes  R_{2k} , \phi\otimes \id  , F \otimes \id )$ on $A- B \otimes \cclif[2k]$

  There is an equivalence between odd Kasparov bimodules on $A-B$ and Kasparov bimodules on $A- B \otimes \cclif[2k+1] $  given by assigning the odd Kasparov bimodule  $( E, \phi, F)$ on $A-B$ to $(E \otimes  R_{2k+1} , \phi\otimes \id  , F \otimes \omega )$ on $A- B \otimes \cclif[2k+1]$
\end{prop}

Now, for $A$ a graded $\cstarr$.
We define a map on the complex $\cstar$  $A_\C = A \otimes \C$ as the identity on $A$ and as conjugation on complex numbers : $\id \otimes \bardot$.
This motivates the following definition
\begin{defi}
  A pair $(A, \overbar \cdot)$ is called a $\cstar$ with real structure if $A$ is a complex $\cstar$ and if $\bardot : A\to A$ is an involutive $*-$antilinear morphism, that is for every $\lambda $ in $\C$ and $a,b$ in $A$ we have:
  \begin{align*}
    \overbar{\lambda a}  &= \overbar \lambda \overbar{\vphantom{\lambda}a} & \overbar{\overbar a} &= a & \overbar{a^*} &= \overbar a ^* & \overbar{a\cdot b}& = \overbar{\vphantom{b} a} \cdot \overbar b
  \end{align*}
  A morphism between $\cstars$ with real structures is a $*-$morphism that intertwines the antilinear involutions.
\end{defi}

One recover a $\cstarr$ $A_\R$  as the set of fixed points of the involution.
If $A$ is a $\cstarr$ we denote by $A_\C= (A\otimes \C, \bardot)$ its $\cstar$ with real structure as defined above.
Such pairs form then a category equivalent to the one of real $\cstarrs$.

Let $\clif$ be a real Clifford algebra.
For a complex graded module $R$ on $\clif_\C$, its conjugate is defined using the real structure on $\cclif$:
as a set $\bar R$ is $R$ and the action of $a \in \clif$ is given by the $R$ action of $\bar a$ and with the same grading.
Take now $R$ one of the irreducible graded module on $\clif_\C$.
Distinguish two cases, depending on the Morita class of $\clif_\C$.

Assume that $\clif_\C$ is a matrix algebra.
The representation $\bar R$ is either isomorphic to $R$ or to $R^o$.
In any case there is an operator $J_R: R \to R$ such that
\[ \forall a \in \clif_\C, \quad \forall x \in R, \quad \bar a x = J_R a J_R^{-1} x\]
The operator $J_R^2$ intertwines $R$ with itself and must be a multiple of the identity.
Multiplying $J_R$ by a real number one can assume that $J_R^2 = \alpha \id$ with $\alpha = \pm 1$.
Depending on whether $\bar R \simeq R$ or $\bar R\simeq R^o$ we have $J_R\epsilon_R = \alpha'' \epsilon_R J_R$ with $\alpha'' = \pm 1$.
The number $\alpha$ does not depend of the graded structure of $\clif$.
It is just the class upon Morita equivalence of the ungraded $\clif$ in $Br(\R) = \{ \left[\R\right],  \left[\H\right]\}$.
Take generators  $c_1,\cdots,  c_{2n}$ of $\clif$.
Then the operator $\nu = c_1 \cdot \cdots \cdot c_{2n}$ squares to ${(-1)}^n c_1^2 \cdots c_{2n}^2$ commutes with $\epsilon_R$ and $J_R$ and anticommutes with each $c_i$.
The operator $\epsilon_R \nu$ intertwines $R$ and $R$ and it then a multiple of the identity that squares to $\pm \id$.
We then have $\alpha'' = {(-1)}^n c_1^2 \cdots c_{2n}^2$.

If, now $\clif_\C$ is the sum of two matrix algebras over $\C$.
We then have $R^o \simeq R$ as graded $\clif_\C$ modules and we have an intertwining operator $\omega$.
This means that $\omega: R \to R$ is $\clif_\C$ linear and verifies $\omega \epsilon_R = - \epsilon_R \omega$.
Changing $\omega$ for some complex  multiple of it we assume that $\omega^2 = \id$.
Now, $\bar R$ is isomorphic to $R$ and we write again $J_R: R \to R$ a unitary intertwining those two representations.
We can assume as before $J_R^2 = \alpha \id$.
We have $\omega J_R\omega$ that is also such an intertwining operator, by Schur lemma we have $J_R \omega = \alpha' J_R \omega $.
In fact we can take
\[\omega = i^{\alpha'}c_1 \cdots c_{2n+1} \quad \text{where} \quad  \alpha' ={(-1)}^n c_1^2 \cdots c_{2n+1}^2\]
With this $\omega$ one could define an additional Clifford action on $R$ by $c = \omega \epsilon_R$.
Note that $c^2 = -\alpha' \id$.
This gives $R$ the structure of a ${(\clif \gotimes \clif[1,0])}_\C$ or ${(\clif \gotimes \clif[0,1])}_\C$ module.
One can then apply the previous discussion to recover the sign $\alpha$ for $\clif$.

Summing up we obtain

\begin{prop}\label{clifrep}
  Let $\clif$ be a real Clifford algebra and let $\clif_\C$ its complexified algebra.
  Let $R$ one of the irreducible graded module on $\clif_\C$.
  There is an antilinear operator $J_R: R \to R$ such that
\[ \forall a \in \clif_\C, \quad \forall x \in R, \quad \bar a x = J_R a J_R^* x \quad J_R^2 x = \alpha x \quad \alpha = \pm 1\]

For $\clif_\C$ a matrix algebra. The two $R$ and $R^o$ are the only irreducible graded modules on $\cclif$.
The sign $\alpha$ is the class of the ungraded $\clif$ in $Br(\R) = \{ \left[\R\right],  \left[\H\right]\}$. Furthermore, we have $J_R\epsilon_R = \alpha'' \epsilon_R J_R$ with $\alpha'' = \pm 1$.

If, now $\clif_\C$ is the sum of two matrix algebras over $\C$.
The graded module $R$ is the only irreducible one on $\cclif$.
There is an operator $\omega: R \to R^o$ of square $\id$ verifying $\omega \epsilon_R = - \epsilon_R \omega$  and $J_R \omega = \alpha' J_R \omega $ for two signs $\alpha$ and $\alpha'$ depending only on the Morita class of $\clif$.
\end{prop}

Signs $\alpha, \alpha'$ and $\alpha''$ will be collected below in~\cref{real8}

\begin{remark}\label{orientation}
  There is an ambiguity in our discussion that we now lift.
  We have to choose in the even case between an irreducible graded module and its opposite and in the odd case between a chirality operator $\omega$ and its opposite  $-\omega$.
  This is done by choosing an \emph{orientation} of the Clifford algebra as in~\cite{k81}, defined as an ordered set of Clifford generators $\ori=  (c_1, \ldots, c_n)$.
  This choice defines uniquely a \emph{chirality operator} by
  \begin{equation}\label{chiror}
\omega_\ori = i^{q+n(n-1)/2}c_1 \cdots c_n     
\end{equation}
This choice depends uniquely on the class of $\ori = (c_1, \cdots, c_n)$ in $\O_n / \SO_n$

We are then left with a chirality operator in the odd case.
In the even case, let $R_\ori$ the unique irreducible $\clif_\C$ module such that the action of $\omega_\ori$ coincides with the grading.

This choice has a good compatibility with Kasparov product.
For two orientations  $\ori=(v_1, \cdots , v_n)$ and $(w_1, \cdots , w_m)$ of two finite dimensional vector spaces denote the orientation on the direct sum $(v_1, \cdots , w_m)$ obtained by concatenation as $\ori \cdot \ori' $,  we have:
 \begin{itemize}
  \item An isomorphism in the even/even case $R_\ori \gotimes R_{\ori'} \simeq R_{\ori \cdot \ori'}$
  \item An equality in the odd/even case $\omega_\ori \gotimes \id = \omega_{\ori \cdot \ori'}$ on $R_\ori \gotimes R$. In the other way round, $\epsilon' \otimes \omega_\ori = \omega_{\ori' \cdot \ori} = - \omega_{\ori \cdot \ori'}$. 
    \item In the the odd/odd case a splitting  $R_{\ori \cdot \ori'}  +R_{\ori \cdot \ori'}^o $ of  $R \gotimes R' $ given by the 1 and $-1$ eigenspaces of  $\omega_\ori\epsilon \otimes \omega_{\ori'}\epsilon'$. 
  \end{itemize}
\end{remark}

\begin{remark}
  The proof differs from the case by case study one can find for example in~\cite{gvf13}.
  It is inspired by~\cite{w64, def99} on a work of Wall on graded Brauer groups of fields.
  We just used the fact that the Clifford algebra $\clif$ was a real simple central graded algebra.
  Apart from the identification of the signs, we could have taken $\clif$ to be any simple central graded algebra over a field $\mathbf{k}$ of characteristic distinct from 2.
  By \emph{semi-simple} we mean that every graded module over it is decomposable as a sum of irreducible ones.
  By \emph{central} we mean that the degree 0 elements commuting with every other elements of the algebra is a multiple of the identity.
  The set $sBr(\mathbf k)$ of Morita classes of such algebras forms a group upon the graded tensor product.
  For $\mathbf{k_s}$ the separable closing of $\mathbf k$, there are only two such algebras and  $sBr(\mathbf{ k}_s) = \ZZ$.
  If $R$ is an irreducible graded module on $A \gotimes \mathbf{k}_s$ then for any $\phi \in Gal(\mathbf{k},  \mathbf{k}_s)$ we can build $R^\phi$ the irreducible graded $A \gotimes \mathbf{k}_s$ module, $R$ twisted by $\phi$.
  One first distinguish when $A \gotimes \mathbf{k}_s$ is a matrix algebra, that is if $[A \otimes \mathbf k_s] = 0 \in sBr(\mathbf k_s) = \ZZ$.
  Then we look if $R^\phi \simeq R$ or $R^\phi \simeq R^o$.
  One then obtain a morphism $Gal(\mathbf{k},  \mathbf{k}_s) \to \ZZ$.
  One can then obtain an isomorphism of Wall between $sBr(\mathbf{k}) $ and an extension of $\ZZ$ by an extension of $\hom(Gal(\mathbf{k},  \mathbf{k}_s), \ZZ) \simeq \mathbf k^* / \mathbf k ^{*2}$ with the ungraded Brauer group $Br(\mathbf{k})$ of $ \mathbf{k}$.
  This gives $sBr(\R) \simeq \Z_8$.
  \end{remark}

Now, for any graded Hilbert $B\gotimes \clif-$module $(E,\epsilon)$, the complexified module $(E_\C, \epsilon_\C$ on $B_\C\otimes \clif_\C$.
We distinguish as before two cases:

If $\clif_\C$ is a matrix algebra,  $E_\C$ can be written in the form $\tilde E \gotimes R \simeq \tilde E_0 \otimes R + \tilde E_1 \otimes R^o$ for $R$ (one of the) irreducible graded module over $\clif_\C$.
Denote  $\tilde E = \tilde E_0 + \tilde E_1$ the $B-$Hilbert module with grading $\tilde \epsilon$.
We can then express on $\tilde E\gotimes R$ the grading  $\epsilon_\C $ as $ \tilde \epsilon \otimes \epsilon_R$.
Thanks to the Schur lemma the conjugation on $E_\C$ can be written $J \otimes J_R$ with $J$ an anti $B_\C-$unitary and  $J_R$ the intertwiner between $\bar R$ and $R$ or $R^o$.
Recall that $J_R^2 = \alpha \id$ and that $J_R \epsilon_R = \alpha'' \epsilon_R J_R$ so that $J^2 = \alpha \id$ and $\tilde \epsilon J = \alpha'' J \tilde \epsilon$.
If now $T$ is in $\endo_B(E)$,  we can write $T_\C$ as $\tilde T \gotimes 1$.
The degree of $\tilde T  $ on $(\tilde E, \tilde \epsilon)$ is then equal to the degree of $T$.
We have $J \tilde T = \tilde T J$.

If, now $\clif_\C$ is the sum of two matrix algebras over $\C$.
There is a unique irreducible graded module on $\clif_\C$ that we denote $R$.
It is isomorphic as a graded $\clif_\C$ module to its opposite module and we have an intertwining unitary $\omega$ realizing this isomorphism.
The $B_\C\gotimes \cclif-$Hilbert graded module $E_\C$ can be written $\tilde E \otimes R$ for $E$ an ungraded $B-$Hilbert module.
In this setting, the conjugation on $E_\C$ can be written as before $J \otimes J_R$ so that $J^2 = \alpha \id$.
If now $T$ in $\endo_B(E)$ is even ,  we can write $T_\C$ as $\tilde T \otimes 1$.
We have $J \tilde T = \tilde T J$.
If now $T$ in $\endo_B(E)$ is odd,  we can write $T_\C$ as $\tilde T \otimes \omega$.
We then have $J \tilde T = \alpha' \tilde T J$.

Application of this to Kasparov bimodules leads to the following definition one can find in~\cite{gvf13} or in a different form in~\cite{c95a} to characterize Higher $\KKO$ groups:
\begin{defi}\label{real8}
  Let $A, B$ be two ungraded $\cstarrs$ and $d $ be a class in $\Z/ 8\Z$.
  A tuple $(E, \phi, F, J)$ is called a real Kasparov bimodule with $\KO$ dimension or degree $d$ on $A-B$ or on $A_\C - B_\C$ if $(E, \phi, F)$ is a complex Kasparov triple on the complex $\cstarrs$ $A_\C$ and $B_\C$ of even or odd degree depending on the class of $d$ in $\ZZ$ and $J: E \to E$ is an anti-$B$-linear operator verifying :
  \begin{align*}
      \forall b \in B_\C, \overbar{b} = JbJ^* ,  &&\forall a \in A_\C,\phi(\overbar{a}) = J\phi(a)J^* ,  && J^2 = \alpha, && FJ = \alpha_d' JF, && J\epsilon = \alpha_d'' \epsilon J
  \end{align*}
  Where signs $\alpha_d, \alpha_d'$ and $\alpha_d''$ depend on the value of $d$ and are given by the following table:

  \begin{center}
    \begin{tabular}{lrrrrrrrr}
      \toprule
      $d$&0&1&2&3&4&5&6&7\\
      \midrule
      $\alpha_d$&$1$&$1$&$-1$&$-1$&$-1$&$-1$&$1$&$1$\\
      $\alpha'_d$&$1$&$-1$&$1$&$1$&$1$&$-1$&$1$&$1$\\
      $\alpha''_d$&$1$& &$-1$& &$1$& &$-1$& \\
      \bottomrule
    \end{tabular}
  \end{center}
  An isomorphism between such Kasparov bimodules is given by a usual isomorphism of complex Kasparov bimodules intertwining the action of the corresponding antilinear operators.
\end{defi}

\begin{remark}
  Note that both $(E, \phi, F, J)$ and $(E, \phi, F, -J)$ are isomorphic through the action of the imaginary $i$.
\end{remark}

\begin{prop}
  Let $A, B$ be two ungraded $\cstars$.
  The set of $A\otimes \clifpq-B$ Kasparov bimodules identifies naturally with the set of $A-B$ Kasparov modules with real structure of $\KO$ dimension $p-q$.

  If $R$ in the (or one of the) irreducible graded $\clifpq_\C$ module then this correspondence is given by
  \[(E, \phi, F) \mapsto (E \otimes R, \phi \otimes 1, F \otimes \id)  \text{  in the even case}\]
  \[(E, \phi, F) \mapsto (E \otimes R, \phi \otimes 1, F \otimes \omega) \text{  in the odd case}\]

\end{prop}

\begin{remark}\label{evengraded}
  For the even case, an other option would be to change $J$ for $J\epsilon$ when $J$ anticommutes with $\epsilon$.
  This is more natural with respect to the graded picture as $F$ graded commutes with $J$ and $J$ is an operator of square $\beta_d = \alpha_d \alpha_d''$ and degree $\beta_d''= \alpha_d''$.
  \begin{center}
    \begin{tabular}{lrrrr}
      \toprule
      $d$&0&2&4&6\\
      \midrule
      $\beta_d$&$1$&$1$&$-1$&$-1$\\
      $\beta_d''$&$1$ &$-1$ &$1$ &$-1$ \\
      \bottomrule
    \end{tabular}
  \end{center}
\end{remark}
This amounts to write on $E_\C = \tilde E \gotimes R$ the conjugation operator as $J \gotimes J_R$.

\begin{remark}\label{oddgraded+}
  When $d = q - p$ is  odd, given a Kasparov bimodule $(E, \phi, F) $ on $ A - B\otimes \clif[p+1,q+1 ]$ we consider the Clifford action as an action of $\clif[p+1, q]$ with an additional unitary $c$ that squares to $-\id$.
Write $E_\C$ as $\tilde E\gotimes R$ where $\tilde E$ is graded and where $R$ represents one of the unique $\clif[p+1, q]$ irreducible module.
This comes with a writing $F_\C = \tilde F \otimes \id$ and an antiunitary $J$ of square and degree given by~\cref{evengraded} for $\KO$ dimension $d-1$  graded commuting with $F$.
Now $ c\epsilon$ can be written $\tilde c \otimes \id$ where $\tilde c$ is an odd operator on $\tilde E$ whose square is $\id$ and graded commutes with $J$.
In fact both the operators $c$, $F$ and $J$ graded commute and sign conditions read $J^2 = \beta_{d-1}$ and $\deg(J) = \beta_{d-1}''$, gathered in the following table:
\begin{center}
  \begin{tabular}{lrrrr}
    \toprule
    $d$&1&3&5&7\\
    \midrule
    $\beta_{d-1}$&$1$&$1$&$-1$&$-1$\\
    $\beta_{d-1}''$&$1$ &$-1$ &$1$ &$-1$ \\
    \bottomrule
  \end{tabular}
\end{center}

To make the link with the previous presentation, we unravel Morita invariance.
The Kasparov module $(E \otimes R, \phi \otimes \id , F_0 \otimes \omega)$ on $A -B\gotimes \clif[p,q]$ leads to the bimodule on $A -B\gotimes \clif[p+1,q+1]$ given  on $(E + E) \gotimes R$ by the grading $(\begin{smallmatrix} \epsilon_R &0 \\ 0 & -\epsilon_R \end{smallmatrix})$ and by the additional Clifford action of $c = (\begin{smallmatrix} 0 & -\epsilon_R \\\epsilon_R & 0 \end{smallmatrix})$ and $(\begin{smallmatrix} 0 & \epsilon_R \\ \epsilon_R & 0 \end{smallmatrix})$.
The operator is given by $(\begin{smallmatrix} F_0 \omega & 0 \\ 0 & -F_0 \omega \end{smallmatrix})$.
Now $ (\begin{smallmatrix} 0 & i\omega \\  -i\omega & 0 \end{smallmatrix})$ commutes with the grading and with the $\clif[p,q+1]$ action.
The two eigenspaces of this operator of square $\id$ denoted $E_1$ and $E_{-1}$  are identified with the $A -B\gotimes \clif[p,q]$ bimodule  $E \otimes R$ by the following isomorphisms:
\begin{align*}
  x \in E \otimes R \mapsto  \left(\begin{matrix} i \omega x \\x \end{matrix}\right) \in E_1 &&   x \in E \otimes R \mapsto  \left(\begin{matrix}x\\i\omega  x \end{matrix}\right) \in E_{-1} 
\end{align*}

As in~\cref{orientation} the action of $i \epsilon_R \omega$ gives $R$ the structure of the irreducible graded module on $\cclif[p,q+1]$ that we denote $R_+$ for which the two previous isomorphisms  $E_1 \simeq E \otimes R_+$ and $E_{-1} \simeq E\otimes R_+^o$ become isomorphisms of $A -B\gotimes \cclif[p,q+1]$ bimodules.
An intertwining operator for the conjugate representation of $R_+$ is given by $J_R\omega$ if $\alpha'_{d}=1$ and otherwise can be chosen to be just $J_R$.
We write $(E +E ) \gotimes R = E_1 + E_{-1} \simeq (E + E) \gotimes R_+$ and we have our Kasparov bimodule of $\KO$ dimension $q-p$ as $(E+E, \phi, F, c, J) )$ with the operators given by
\begin{align*}
    \epsilon &=\left( \begin{matrix}1 & 0 \\ 0 & -1 \end{matrix}\right)&
      \phi &=\left( \begin{matrix}\phi_0 & 0 \\ 0 & \phi_0 \end{matrix}\right)&
  F = &\left(\begin{matrix}0 & -iF_0 \\i F_0 & 0\end{matrix}\right) &
  c &= \left(\begin{matrix}0 & 1 \\ 1 & 0 \end{matrix}\right)&
  J &= \left(\begin{matrix}0 & -J_0 \\ J_0 & 0 \end{matrix}\right)
  \text{ or }  \left(\begin{matrix} J_0 & 0  \\ 0& J_0  \end{matrix}\right)
\end{align*}
The operators $J$ of degree $\alpha'_{d}$ then squares to $-\alpha'_{d}J_0^2 = -\alpha_d \alpha'_d$ which gives also the sign $\beta_{d-1}= \alpha_{d-1} \alpha''_{d-1}$.
\end{remark}

\begin{remark}\label{oddgraded-}
The $\clif[p+1, q+1]$ action could be interpreted as a $\clif[p,q+1]$ and an additional Clifford action of an element of square $\id$.
That is a quintuple $(E, \phi, F, J, c)$ such that $(E, \phi, F)$ is a complex Kasparov bimodule on $A- B$, $J$ is an antilinear operator on $E$ compatible with the real structures of $A$ and $B$, of degree $\beta_{d+1}''$ and square $\beta_{d+1}$ graded commuting with $F$ and with $c$ of square $-1$ and degree 1 graded commuting with both $J$ and $F$
\begin{center}
  \begin{tabular}{lrrrr}
    \toprule
    $d$&1&3&5&7\\
    \midrule
    $\beta_{d+1}$&$1$&$-1$&$-1$&$1$\\
    $\beta''_{d+1}$&$-1$&$1$ &$-1$ &$1$  \\
    \bottomrule
  \end{tabular}
\end{center}
As before, an odd Kasparov module $(E, \phi_0 , F_0)$ given as in~\cref{real8} yields

\begin{align*}
    \epsilon &=\left( \begin{matrix}1 & 0 \\ 0 & -1 \end{matrix}\right)&
      \phi &=\left( \begin{matrix}\phi_0 & 0 \\ 0 & \phi_0 \end{matrix}\right)&
  F = &\left(\begin{matrix}0 & F_0 \\F_0 & 0\end{matrix}\right) &
  c &= \left(\begin{matrix}0 & -1 \\ 1 & 0 \end{matrix}\right)&
  J &= \left(\begin{matrix} J_0 & 0 \\ 0 & J_0 \end{matrix}\right)
  \text{ or }  \left(\begin{matrix} 0 &J_0  \\- J_0&0  \end{matrix}\right)
\end{align*}
\end{remark}

\subsection{The Kasparov product}

The Kasparov product can be made explicit in this setting using the characterization of~\cite{cs84}.
As this construction involves graded tensor products we use the graded formulation of real structures of~\cref{evengraded,oddgraded+}:

\begin{prop}\label{kaspairreal}
  Let $A$, $B$, $D$, be three ungraded $\cstarrs$ with $A$ separable and $B$ and $D$ $\sigma-$unital.
  Let two elements $[E_1, \phi_1, F_1, J_1] \in \KKO_n(A, B)$ and $[E_2, \phi_2, F_2, J_2]$ given by the $\beta$ formulation and with eventually a supplementary action of an odd square $\id$ or $-\id$ operator.
  The Kasparov product of those two elements over $D$ is given in the $\beta$ formulation by:
  \begin{description}
  \item[if $n$ and $m$ are even]$[E , \phi , F, J ] =  [ E_1\gotimes_D E_2,\phi_1 \otimes \phi_2, F, J_1 \gotimes J_2]$ for any odd operator $F: E \to E$ verifying~\cref{real8} and
  \begin{itemize}
    \item $\forall x \in E_1$ homogeneous the map $ y  \mapsto x \gotimes F_2(y) - {(-1)}^{\partial x}F(x\gotimes y)$ is a compact operator from $E_2$ to $E$
    \item $\forall a \in A, \phi(a)[F_1 \gotimes \id, F] \phi(a)^* \geq 0 \mod \compact(E)$
  \end{itemize}
  When $A = \R$ , $E_1$ is given by some $H_D$ and $F_1$ commutes with the action of $\R$ is self adjoint, verifies $\norm{F_1} \leq 1$ such an operator on $H \gotimes E_2$ is homotopic to
  \[F_1 \gotimes \id + (\id \otimes \phi_2)(\sqrt{1-F_1^2}) F_2 \]

  \item[if $n$ is odd and $m$ is even] $[E , \phi , F, J,c ] = [ E_1\gotimes_D E_2,\phi_1 \otimes \phi_2, F, J_1 \gotimes J_2, c_1 \gotimes \id]$ where $F$ satisfies the same conditions as before.

  \item[if $n$ is even and $m$ odd]$[E , \phi , F, J,c ] = [ E_1\gotimes_D E_2,\phi_1 \otimes \phi_2, F, J_1 \gotimes J_2, \id \gotimes c_2]$ where $F$ satisfies the same conditions as before.

  \item[if $n$ and $m$ are odd]
  $[E , \phi , F, J ] $ where $E$ is the image of the even projector $p = (\id + c_1 \gotimes c_2)/2$.
  The restriction of $\phi_1 \otimes \phi_2$ on this space gives $\phi$ and the restriction of $J_1 \gotimes J_2$ gives $J$.
  The operator $F: E \to E$ is any odd operator verifying ~\cref{real8} and
  \begin{itemize}
    \item $\forall x \in E_1$ homogeneous the map $ y  \mapsto p(x \gotimes F_2(y)) - {(-1)}^{\partial x}F(p(x\gotimes y))$ is a compact operator from $E_2$ to $E$
    \item $\forall a \in A, \phi(a)[(F_1 \gotimes \id)p, F] \phi(a)^* \geq 0 \mod \compact(E)$
  \end{itemize}
  When $A = \R$, $F_1$ commutes with the action of $\R$ and $\norm{F_1} \leq 1$ such an operator is homotopic to
  \[F_1 \otimes_D \epsilon_2 +  \sqrt{1 -F_1^2} \otimes_D F_2  \]
  \end{description}

\end{prop}

\begin{proof}
    For two Kasparov bimodules $(E_1, \phi_1, F_1)$, $(E_1, \phi_2, F_2)$ are real Kasparov bimodules over $A-D\otimes \clifpq$ and over $D- B\otimes \clif[p',q']$
  Writing $E_{1, \C} = \tilde E_1 \gotimes R_{p,q}$ and $E_{2,\C} = \tilde   E_2 \gotimes R_{p',q'}$, we have the isomorphism of $A-B \otimes \clif[p +p', q+q']$ bimodules:
  \[E_{1,\C} \gotimes_D E_{2,\C} = \tilde E_1 \gotimes \tilde E_2 \gotimes R_{p,q} \gotimes R_{p',q'}\]
  We then use~\cref{chiror} as to identify $ R_{p+p',q+q'}$ with $R_{p,q} \gotimes R_{p',q'} $ if at least one of the degree is even and with the invariant subspace of $c\otimes c'$ if both degrees are even.  
It is then just a matter of translation in this framework of the connection characterization of Connes and Skandalis.
\end{proof}

\begin{remark}
  Note that the same formulas can be given when the odd Kasparov triples are given with the data of an odd self adjoint unitary squaring to $-1$ as in~\cref{oddgraded-} 
\end{remark}

\subsection{The exact sequence via real structures}
In this section we investigate morphisms of the exact sequence~\cref{kwprop} in the context of higher $\KKO$ groups with real structures of~\cref{secrealkko}.
Let $A$ and $B$ be two ungraded $\cstarrs$
      Complexification in this framework just consists in forgetting the real structure:
      \begin{prop}
        Complexifying a Kasparov bimodule with real structure $(E, \phi, F, J)$ on $A_\C - B_\C$ gives the underlying complex Kasparov bimodule $(E, \phi, F)$.
      \end{prop}

      The action of $\eta$ we will be described by the graded $\beta$ formalism presented in~\cref{evengraded} in the even case and in~\cref{oddgraded+,oddgraded-} in the odd case.
      We then link back to the $\alpha$ formalism of~\cref{real8}.
      Starting with the action of $\eta$ on even Kasparov triples we have:

      \begin{lemma}
        Let $(E, \phi, F, J)$ be an even degree Kasparov bimodule in the graded presentation of~\cref{evengraded}.
      The action of $\eta$ on this triple is given in the $\beta^+$ presentation of~\cref{oddgraded+} by the following  quintuple $\left(E + E^o,
  \left(\begin{smallmatrix}\phi & 0 \\ 0 & \phi \end{smallmatrix}\right),
  \left(\begin{smallmatrix} \vphantom{\phi}F & 0 \\ \vphantom{\phi}0 & -F \end{smallmatrix}\right),
  \left(\begin{smallmatrix}J \vphantom{\phi}& 0 \\ 0\vphantom{\phi} &  \pm J \end{smallmatrix}\right),
  \left(\begin{smallmatrix}0 & 1\vphantom{\phi} \\ \vphantom{\phi}1 & 0 \end{smallmatrix}\right)\right)$ where the sign is the degree $\beta''$ of the operator $J$.
\end{lemma}
\begin{proof}
The module $E + E^o$ is isomorphic to $E \otimes \C^2$ where the grading is only supported on $\C^2$.
An isomorphism can be given by $\frac{1+ \epsilon}{2}\left(\begin{smallmatrix} 1 &0 \\0 & 1 \end{smallmatrix} \right) + \frac{1- \epsilon}{2}\left(\begin{smallmatrix} 0 &1  \\ 1 & 0 \end{smallmatrix} \right)$.
The operators are then given by $
  \left(\begin{smallmatrix}\phi & 0 \\ 0 & \phi \end{smallmatrix}\right),
  \left(\begin{smallmatrix} \vphantom{\phi}0 & F\epsilon \\ \vphantom{\phi} -F\epsilon & 0 \end{smallmatrix}\right),
  \left(\begin{smallmatrix} \vphantom{\phi} 0 & J \epsilon \\ -J \epsilon\vphantom{\phi} &  0 \end{smallmatrix}\right) \text{ or }
    \left(\begin{smallmatrix}J \vphantom{\phi}& 0 \\ 0\vphantom{\phi} &  J \end{smallmatrix}\right)$ and $
    \left(\begin{smallmatrix}0 & 1\vphantom{\phi} \\ \vphantom{\phi}1 & 0 \end{smallmatrix}\right)$
    \end{proof}
Now linking back to the $\alpha$ presentation of~\cref{real8} :
      \begin{prop}\label{eveneta}
        Let $(E,\epsilon, \phi, F, J)$ be a Kasparov bimodule with real structure of even $\KO$ dimension. The action of $\eta$ gives the Kasparov bimodule $(E, \phi, iF\epsilon, J)$ with real structure of odd $\KO$ dimension.
      \end{prop}

      \begin{proof}
        The previous form of the operator is of the same shape as the one obtained in~\cref{oddgraded+} where were obtained the link between $\alpha$ and $\beta^+$ formalism.
        Changing eventually our antilinear operator for $J\epsilon$ to always have the commutation of it in the non graded sense with the Fredholm operator we also put the initial even  Kasparov triple in an $\alpha$ form.  
      \end{proof}
    For odd Kasparov triples now :
\begin{lemma}
  Let $(E, \phi, F, J, c)$ be a Kasparov bimodule given in the $\beta^-$ presentation of~\cref{oddgraded-} with an odd $c$ of square $-\id_E$ graded commuting with $F$ and $J$.
  Applying $\eta$ to this bimodule gives the same bimodule in the $\beta$ graded formalism  but without $c$:
  \[\eta (E, \phi, F, J, c) = (E, \phi, F, J)\]
\end{lemma}

\begin{proof}
  The element $\eta = j_* \in  \KKO(\R, \clif[0,1])$ for the inclusion $j: \R \to \clif[0,1]$ gives an element in $\KKO(\clif[1,0], \R)$ given by the action of $( \begin{smallmatrix}
    0 & 1 \\ 1 & 0
  \end{smallmatrix})$ on $\R^2$.
  We can then interpret the action of $\eta$ as a forgetful functor for this $\clif[1,0]$ action.
  But the $\beta^-$ formalism is precisely given by distinguishing such an action as to give the operator $c$.
\end{proof}

  \begin{prop}\label{oddeta}
    Let $(E, \phi, F, J)$ be a Kasparov bimodule with real structure of odd degree $d$ given in the $\alpha$ formalism of~\cref{real8} by an antiunitary $J: E \to E$ and no grading on $E$. The graded tensor product of this bimodule with $\eta$ gives the Kasparov bimodule $(E + E,
    (\begin{smallmatrix} 1 & 0 \\ 0 & -1 \end{smallmatrix}),
    (\begin{smallmatrix} \phi & 0 \\ 0 & \phi \end{smallmatrix}),
    (\begin{smallmatrix} 0 & F \\ F & 0 \end{smallmatrix}))$ with real structure of even degree given by $  (\begin{smallmatrix} J & 0 \\ 0 & J \end{smallmatrix})$ if $d = 3, 7$ or $  (\begin{smallmatrix} 0 & -J \\ J & 0 \end{smallmatrix})$ if $d = 1, 5$
  \end{prop}

  Here the antilinear operator commutes with the Fredholm operator, for a graded commutation as in~\cref{evengraded} we must change it for  $  (\begin{smallmatrix} 0 & J \\ J & 0 \end{smallmatrix})$ in the case $d= 1, 5$.

Coming now to realification, starting with even dimensional Kasparov triples in $\alpha$ formalism:

  \begin{prop}
    For any complex Kasparov bimodule of even degree $(E , \epsilon, \phi, F) $ on $A_\C- B_\C$.
    Its realification gives the Kasparov module $(E \oplus \bar E, (\begin{smallmatrix} \epsilon_H & 0 \\ 0 & \alpha'' \epsilon_H
   \end{smallmatrix}), (\begin{smallmatrix} \phi & 0 \\ 0 & \phi
   \end{smallmatrix}), (\begin{smallmatrix} F_0 & 0 \\ 0 & F_0
   \end{smallmatrix}))$ with real structure given by $J = (\begin{smallmatrix} 0 & \alpha \\ 1 & 0
    \end{smallmatrix})$.
  \end{prop}

  \begin{proof}
   We follow the chain of isomorphism of graded $A_\C - \cclifpq \otimes B_\C$ bimodules \[(E \gotimes_\C R) \otimes_\R \C \totxt{h} E \gotimes_\C R \oplus \bar E \gotimes_\C \bar R \totxt{\id\otimes \id \oplus \id \otimes J_E} (E \oplus \bar E) \gotimes R \text{ or } (E \oplus \bar E^o) \gotimes R \]

  The graded structure on $E \oplus \bar E$ depending whether $J_R$ intertwines $\bar R$ with the irreducible module $R$ on $\cclifpq$ or its opposite $R^o$.

  The conjugation operator is given by $ (\begin{smallmatrix} 0 & \alpha \\ 1 & 0
  \end{smallmatrix}) \otimes J_R$ and the Fredholm operator by $(\begin{smallmatrix} F_0 & 0 \\ 0 & F_0
  \end{smallmatrix}) \otimes \id$. 
  \end{proof}

  In the odd case now, from odd complex Kasparov bimodules:

    \begin{prop}
    For any complex Kasparov bimodule of odd degree $(E ,  \phi, F) $ on $A_\C- B_\C$.
    The image of it by the realification morphism is the Kasparov bimodule $(E \oplus \bar E,  \phi, F, J)$ with
    \begin{align*}
      \phi &= \left(\begin{matrix} \phi & 0 \\ 0 & \phi
     \end{matrix}\right)&
     F &=  \left(\begin{matrix} F & 0 \\ 0 & \alpha' F
     \end{matrix}\right) &
     J &= \left(\begin{matrix} 0 & \alpha \\ 1 & 0
      \end{matrix}\right)
    \end{align*}
  \end{prop}

  \begin{proof}
  Any Kasparov bimodule $(E \otimes R, \phi\otimes \id, F \otimes \omega)$ on $A - B \otimes \cclifpq$ gives the chain of isomorphisms:
  \[(E \otimes_\C R) \otimes_\R \C \totxt{h} E \otimes_\C R \oplus \bar E \otimes_\C \bar R \totxt{\id\otimes \id \oplus \id \otimes J_R} E \otimes R \oplus \bar E \otimes R \simeq (E \oplus \bar E) \otimes R\]
  Where   $h(x \otimes \lambda) = (zx, \bar z x)$ for $x \in E \gotimes_\C R$ and $\lambda \in \C$.
  The inverse of $h$ is given by $h^{-1}(x, y) = \frac{1}{2}(x + y ) \otimes 1  + \frac{i}{2}( y - x ) \otimes i$.
  Then $(x, y)$ in $ E\otimes R \oplus \bar E \otimes R$ translates into $\frac{1}{2}(x + J_R(y) ) \otimes 1  + \frac{i}{2}( J_R(y) - x) \otimes i$ in $E \otimes_\C R \otimes_\R \C $
applying $F \otimes w \otimes 1 $ we obtain $\frac{1}{2}(F(x) + F(J_R(y)) ) \otimes 1  + \frac{i}{2}( F(J_R(y)) - F(x) ) \otimes i$.
Conjugation on $\C$ gives:
\begin {align*}
\frac{1}{2}(x + J_R(y) ) \otimes 1  - \frac{i}{2}( J_R(y) - x ) \otimes i
  &= \frac{1}{2}(J_R(y) + \alpha J_R(J_R(x)) ) \otimes 1  + \frac{i}{2}( \alpha J_R(J_R(x)) - J_R(y) ) \otimes i
\end{align*}
  \end{proof}

\subsection{Finite dimensional \texorpdfstring{$\KO$}{KO} theory via symmetries}
Let $A$ be a trivially graded $\cstarr$ and $A_\C$ its complex $\cstar$ with real structure.
As explained in~\cref{ko} an other description of $\KO$ theory due to Karoubi is given by equivalences classes of unitaries in some matrix algebra.
When $A$ is unital an element of the group $\kok_{p-q}(A)$ is given by the class of a tuple $(s_1, s_2, E)$ where $E$ is a $\clifpq$ right $A-$linear action on some free and finitely generated $A$ module with generators anticommuting with two $A-$linear symmetries $s_1$ and $s_2$.
If $A$ is not unital then we change $A$ for its point unitalization $A^+$ and furthermore assume that $s_1$ equals $s_2$ mod $A$.
In the spirit of~\cref{cliflr} where Clifford algebra where taken from the left to the right as arguments in $\KKO$, consider $E$ as an ungraded left $\clif[q,p]$ module.
In addition $E$ can be considered as an action of $\clif[p,q+1]$ on a graded module by Morita invariance.
Write $\epsilon$ this grading.
Doing so, we can identify $s_1$ and $s_2$ as some matrix over $A \otimes \clif[q+1, p]$.
We then obtain a formulation of $\KO$ theory similar to the one of Van Daele~\cite{v88a, v88} where a $\KO$ cycle of degree $p-q$ is given by cycles $(C, \epsilon,s_1, s_2)$ where
\begin{itemize}
\item $C$ is a representation $\clif[q+1, p] \to \M_n(A)$
\item $\epsilon$ is a self-adjoint unitary in $\M_n(A)$ anticommuting with generators of the Clifford action
  \item $s_1$ and $s_2$ are self adjoint unitaries in $\M_n(A)$ commuting with generators of the Clifford action and anticommuting with $\epsilon$  
\end{itemize}

The main difference with the original Van Daele formulation is the choice of a base point in the set of such operators.
This has another advantage when we will later consider pairings.

We write the complexification $E_\C$ as some $A_\C^n \otimes R$ or $(A_\C^n+ A^m_\C) \otimes R$ as in the previous section.
Any complexification $s_\C$ of a symmetry $s$ acting on $E$, anticommuting with some Clifford action on $E$ can be given a form we describe now in the following proposition:
\begin{prop}\label{finiterealsec}
  Let $d= p-q$.  on $A_\C^n \otimes R$ the conjugation operator is given as $J \otimes J_R$ and  $s_\C$ can be written 
  
  If $d$ is even, as $u \otimes \omega$ where $u =2p-1$ is a symmetry verifying the following relations for $d$ and $\omega$ is the chirality operator of~\cref{orientation}.
  
  If $d$ is odd,  as $ \left(  \begin{smallmatrix}
      0 & u^* \\ u & 0
    \end{smallmatrix}\right)\otimes 1$ with $u$ a unitary of $\M_n(A_\C)$ verifying the following relations :
  \begin{center}
  \begin{tabular}{lrrr}\label{finitereal8}
    \\ \toprule
    $d$ & $J$ anti-unitary & unitary picture & projection picture \\ \midrule
    $0$ & $J^2 = 1$ & $u^* = u, J uJ^* = u $ & $J p J^* = p$ \\
    $1$ & $J^2 = 1$ & $J uJ^* = u $ & \\
    $2$ & $J^2 = 1$ & $u^* = u, J uJ^* = -u $ & $J p J^* = 1-p$ \\
    $3$ & $J^2 = -1$ & $J uJ^* = u^* $ & \\
    $4$ & $J^2 = -1$ & $u^* = u, J uJ^* = u $ & $J p J^* = p$ \\
    $5$ & $J^2 = -1$ & $J uJ^* = u $ &  \\
    $6$ & $J^2 = -1$ & $u^* = u, J uJ^* = -u $ & $J p J^* = 1-p$ \\
    $7$ & $J^2 = 1$ & $J uJ^* = u^* $ & \\
    \bottomrule
  \end{tabular}
\end{center}
\end{prop}

We recover the presentation of $\KO$ theory of unital and trivially graded $\cstarrs$ via projection or unitaries with symmetries one can find for example in \cite{bl15}.

\begin{proof}
  Assume for simplicity that $A$ is unital.
  Schur lemma enables to write the conjugation on $E \otimes R$ as $J\otimes J_R$ where $J$ is a antiunitary of $A_\C^n$ and $s_\C$ as $\tilde s \otimes 1 $ or $(2p-1) \otimes 1$.
   We then transpose the relation of~\cref{clifrep} for the degree $d+1$ to find the result.
  
   In the odd case  $A^n_\C$ comes with a grading.
   We write the symmetry $s_{\C}$ on $A_\C^n\gotimes R$ as $ \left(  \begin{smallmatrix}
      0 & u^* \\ u & 0
    \end{smallmatrix}\right) \otimes 1$ and the conjugation operator as  $\left(
    \begin{smallmatrix}
      J &0 \\0 &J'
    \end{smallmatrix}\right)\otimes J_R$ or  $\left(
    \begin{smallmatrix}
      0 & J \\ -J & 0
    \end{smallmatrix}\right) \otimes J_R$.
The relations given by~\cref{clifrep} between $u_1, u_2$ and $J$ now translates into the given ones in the table. 
\end{proof}

Writing any odd element $\left[ \left(  \begin{smallmatrix}
      0 & u^* \\ u & 0
    \end{smallmatrix}\right) \otimes 1,  \left(  \begin{smallmatrix}
      0 & 1 \\ 1\vphantom{u^*} & 0
    \end{smallmatrix}\right) \otimes 1\right]$ as $[u, J]$ or simply $[u]$ if the antilinear operator is understood, every odd $\kok$ element is given by a difference $[u, J] - [v, J']$.
  Note that for $d=1, 5$ the law on classes of such unitaries coincides with the usual multiplication of matrices.
  For $d = 3, 7$ the additive law is described by $[u_1, J] + [u_2, J] = [u_1 u_2, u_2J u_2^*]$ .
  For $d= 0$ or 4, any $\kok$ class is given by some $\left[ (1-2p) \otimes \omega, \id \otimes \omega \right] - \left[ (1-2q) \otimes \omega, \id \otimes \omega \right] $ that we write $ [p, J] - [q, J']$.
For $d= 2$ or $6$ we can choose an other base point and write any $\kok$ class as  
  $\left[ (1-2p) \otimes \omega, \left(
    \begin{smallmatrix}
      0 & -i \\i & 0
    \end{smallmatrix} \right) \otimes \omega \right] - \left[ (1-2q) \otimes \omega, \left(
    \begin{smallmatrix}
      0 & -i \\i & 0 \end{smallmatrix}\right)\otimes \omega \right] $ written as $ [p, J] - [q, J']$.

\begin{prop}
  The product with $\eta$ is given respectively in even and odd degree by
\begin{align*}
[p, J] &&\mapsto&  \left[1-2p, J \right] \\  
  [u, J] &&\mapsto &\left[
                     \frac{1}{2}\left(\begin{matrix} 1 & u^* \\ u & 1\end{matrix}\right),
                     \left(\begin{matrix} J & 0 \\ 0 & -J \end{matrix}\right) \right] \text{ for }d = 3, 7
         \left[ \frac{1}{2}\left(\begin{matrix} 1 & -iu^* \\ iu & 1\end{matrix}\right) ,
            \left(\begin{matrix}  0 &J     \\ -   J & 0\end{matrix}\right) \right]              \text{ for }  d = 1, 5
\end{align*}
\end{prop}

\begin{proof}
  To an even class represented by a projector $p$ is associated the triple $[1\otimes \omega, (1-2p)\otimes \omega, A_\C^n \otimes R]$.
  Now, as in~\cref{oddeta}, $\eta$, interpreted as $j_*$ for $j$ an inclusion $\R \incl \clif[0,1]$ acts on our element as to give $ \left[\left(  \begin{smallmatrix}
      0 & \vphantom{u^*}1 \\ 1 & 0
    \end{smallmatrix}\right)\otimes 1,
  \left(\begin{smallmatrix}
      0 & 1-2p \\ 1-2p & 0
    \end{smallmatrix}\right)\otimes 1, A_\C^{2n} \otimes R \right]$.

Starting with an odd degree element $\left[ 
  \left(\begin{smallmatrix} 0 & u^* \\ u & 0  \end{smallmatrix}\right)\otimes 1, \left(  \begin{smallmatrix} 0 & \vphantom{u^*}1 \\ 1 & 0  \end{smallmatrix}\right)\otimes 1,
  A_\C^n \otimes R \right]$
 where conjugation in given by $\left(
    \begin{smallmatrix}
      J &0 \\0 &J'
    \end{smallmatrix}\right)\otimes J_R$ or  $\left(
    \begin{smallmatrix}
      0 & J \\ -J & 0
    \end{smallmatrix}\right) \otimes J_R$.
  As in the proof of~\cref{eveneta}, $\eta$ acts as to give
$[
  \left(\begin{smallmatrix} 0 & -iu^* \\ iu & 0  \end{smallmatrix}\right)\otimes \omega,
   \left(  \begin{smallmatrix} 0 & \vphantom{u^*}-i \\ i & 0  \end{smallmatrix}\right)\otimes \omega,
   A_\C^{n} \otimes R ]$.
   Now, it is of the form $\left[ \frac{1}{2}\left(\begin{smallmatrix} 1 & -iu^* \\ iu & 1  \end{smallmatrix}\right), \left(
    \begin{smallmatrix}
      0 & J \\ -J & 0
    \end{smallmatrix}\right) \otimes J_R \right]$
as defined above in dimension 2 and 6.
   In other dimensions, we put it in the form $\left[ \frac{1}{2}\left(\begin{smallmatrix} 1 & u^* \\ u & 1  \end{smallmatrix}\right), \left(
    \begin{smallmatrix}
       J & 0 \\  0 & -J
    \end{smallmatrix}\right) \otimes J_R \right]$ after conjugation with $\left(
    \begin{smallmatrix}
      1 & 0 \\ 0 & i
    \end{smallmatrix}\right)$
\end{proof}
 
As for the link with complex $\K$ theory represented in degree 0 by projector and in degree 1 by unitaries :

\begin{prop}\label{realificationfinite}
  Complexification is given by omission of the real structure :
  \begin{align*}
[p, J] \mapsto&  \left[p\right] &  
  [u, J] \mapsto &\left[u \right]
  \end{align*}
  Realification is given by
   \begin{align*}
     [p] \mapsto&
                  \begin{cases}
                              \left[p\oplus \bar p, J\right]\\
                    \left[p\oplus \overbar{1-p}, J\right]   
                  \end{cases}
     &
  [u] \mapsto &  \begin{cases}
                             \left[u \oplus \bar u, J \right]\\
                    \left[u\oplus \bar u^*, J\right]   
                  \end{cases}
   \end{align*}
   Where $J$ is given by $\left( \begin{matrix}
     0 & \pm \bar \cdot \\  \bar \cdot & 0  
   \end{matrix}\right)
$ accordingly to the degree $d$ of the target group $\KO_d(A)$ and~\cref{finitereal8}
\end{prop}

\section{Applications}
\subsection{\texorpdfstring{$\ZZ$}{Z/2Z} pairing in \texorpdfstring{$\KO$}{KO} theory}
In this section we use the exact sequence~\cref{kwprop} to give formulas for the $\ZZ$ pairings $\left<\KO^n (A) , \KO_{n+1}(A)\right>$ or $\left<\KO^n (A) , \KO_{n+2}(A)\right>$ restricting ourselves to elements in the kernel of complexification $c$ and $\eta$.
Such a restriction for an element in $\KO_{n+1}(A)$ resp. $\KO_{n+2}$ gives a lift to  $\KO_{n}(A)$ resp. $\KU_{n}(A)_\C$ and an integer after pairing with an element in $\KO^n(A)$.
The class in $\ZZ$ of this integer gives the pairing between the original elements.
This is summed up by the following commutative diagram where $\bm \cdot F$ represent the pairing with an element in $\KO^n(A)$:

  \begin{equation*}
   \begin{tikzcd}[row sep = small, column sep = tiny]
      \KO_{n}(A) \arrow[r]\arrow[d, "\bm{\cdot}F"]       &
      \KO_{n+1}(A) \arrow[r]\arrow[d, "\bm{\cdot}F"]			&
      \KU_{n+1}(A_\C)  \arrow[d, "\bm{\cdot}F_\C"]       \\
      \Z \arrow[r]                                       &
      \ZZ \arrow[r]			                                &
      0
    \end{tikzcd}
   \quad \quad
   \begin{tikzcd}[row sep = small, column sep = tiny]
      \KU_{n+2}(A) \arrow[r]\arrow[d, "\bm{\cdot}F"]       &
      \KO_{n+2}(A) \arrow[r]\arrow[d, "\bm{\cdot}F"]			&
      \KO_{n+3}(A_\C)  \arrow[d, "\bm{\cdot}F_\C"]       \\
      \Z \arrow[r]                                       &
      \ZZ \arrow[r]			                                &
      0
   \end{tikzcd}
     \end{equation*}

     When the vanishing of $\eta x $ or $x_\C$ is given by an homotopy, using the Puppe sequence presented along the proof of the exact sequence, one can build a lift $\tilde x$ upon realification, resp. $\eta$ as $\KO$ classes of the corresponding suspension algebra $S^{0,1}A$ or $SA_\C$.
     Now $\bm \alpha_{0,1}F$ gives a class in $\KO^{n+1}(S^{0,1}A)$ and its complexification in $\KU^{n+1}(SA_\C)$.
     We then use Chern Connes pairing formulas to obtain the pairing $\left< \bm \alpha_{0,1}F, \tilde x \right>$.
 Recall first how to build lifts in the exact sequence when given corresponding homotopies: 
     \begin{prop}
       Let $(E, \phi, F)$ a Kasparov triple with an homotopy between $j_*(E, \phi, F)$ and a degenerate module on $A-B\gotimes \clif[0,1]$ represented as a Kasparov triple $(E_t, \phi_t, F_t)$ on $A-SB\gotimes \clif[0,1]$.
       A lift of $[E, \phi, F] \in \KKO_n(A, B)$ to $\KKU_n(A, B)$ is given by the class of the triple $(\tilde E, \tilde \phi, \tilde F)$ on $A_\C-SB_\C\gotimes \cclif[1]$ where:
          \begin{align*}
         \tilde E_t = &E_\C \gotimes \cclif[1] &  \tilde \phi_t =& \phi_\C \gotimes 1  &       \tilde F_t = &\sin_{\pi t / 2} 1 \gotimes c + \cos_{\pi t/2} F_\C\gotimes 1&&& \text{ for } -1 \leq &t \leq 0 \\
          \tilde E_t =& E_{t, \C} & \tilde \phi_t =& \phi_\C  &\tilde F_t= &   F_{t, \C}  &&&\text{for } 0 \leq &t \leq 1   \\
       \end{align*}
     \end{prop}
     \begin{proof}
       Apply $\beta_{0,1}$ to obtain a triple on $A-S^{0,1}B\gotimes \clif[0,1]$ and use the homotopy to obtain  the corresponding triple on $A_\C - SB_\C \gotimes \cclif[1]$ via the Puppe sequence.
     \end{proof}

     \begin{prop}
       Let $(E, \phi, F)$ a Kasparov triple with an homotopy between $(E_\C, \phi_\C, F_\C)$ and a degenerate module on $A_\C-B_\C$ represented as a Kasparov triple $(E_t, \phi_t, F_t)$ on $A_C-SB_\C$.
       A lift of $[E, \phi, F] \in \KKO_n(A, B)$ to $\KKO_{n-1}(A, B)$ is given by the class of the triple $(\tilde E, \tilde \phi, \tilde F)$ on $A-S^{0,1}B$ where:
      \begin{align*}
         \tilde E_t = &E_{-t} &  \tilde \phi_t =& \phi_{-t} & \tilde F_t = &F_{-t} &&\text{for }& -1 \leq &t \leq 0 \\
         \tilde E_t =& \overbar{E_{t}}  &\tilde \phi_t =& \overbar{\phi_t} &\tilde F_t= &   \overbar{F_t}   &&\text{for } &0 \leq &t \leq 1  
       \end{align*}
     \end{prop}
     
     For $\kok$ theory classes given as some finite dimensional Karoubi elements, Bott periodicity is constructed by assigning to a symmetry $s$ in $\M_n(A)$ the matrix $\beta s$ with coefficients in $S^{0,1}A\gotimes \clif[0,1]$ defined as
     \begin{align}\label{fbott}
        \beta s_t = \cos_{\pi t/2} s \gotimes \id + i \sin_{\pi t/2} \id \gotimes e && -1 \leq t \leq 1
     \end{align}
 Bott periodicity is then  given in~\cite{k08a} as the map
       $(s_1, s_2) \mapsto (\beta s_1, \beta s_2)$.
       
       If $j_*[C, s_1, s_2]$ vanishes, we have an homotopy $s_t, 0 \leq t \leq 1$ between $j_*(s_1)$ and $j_*(s_2)$ giving a lift upon complexification in $\KU(A)$ as the class of $(\tilde s_1, \tilde s_2)$ with coefficients in the suspension algebra $SA_\C\gotimes \cclif[1]$ given by
       \begin{align}\label{fheta}         
         \tilde s_{1, t} =&  \cos_{\pi t/2} s_1 \otimes \id + i \sin_{\pi t/2} \id \otimes e &
                                                                                              \tilde s_{2, t} =&    \cos_{\pi t/2} s_2 \otimes \id + i \sin_{\pi t/2} \id \otimes e & -1 \leq t \leq 0   \\
\nonumber          =& s_t & =& j_*(s_2) & 0 \leq t \leq 1
       \end{align}

       If $[s_1, s_2]_\C$ vanishes, we have an homotopy $s_t, 0 \leq t \leq 1$ between $s_{1, \C}$ and $s_{2,\C}$ giving a lift upon $\eta$ in as the class of matrices $(\tilde s_1, \tilde s_2)$ with coefficients in the suspension algebra $S^{0,1}A$ given by   
       \begin{align}\label{fhc}
          \tilde s_{1, t} = \overbar{s_{-t}}\text{ and } \tilde s_{2, t} =   s_2  &\text{ if }  -1 \leq t \leq 1 & \tilde s_{1, t} = s_t \text{ and } \tilde s_{2, t} = s_2 &\text{ if } 0 \leq t \leq 1
       \end{align}

       Following~\cite{c85} an element $a$ in $A$ is said to be \emph{$p$ summable} with respect to the Fredholm module $(H, \phi, F)$ over $ A$ when  $[F, \phi( a)] $ is in $ \summable^p $  the ideal of $p$ summable operators of $H$.
Denote the $*-$algebra of such $p$ summable elements $\summable^p_F$.
It is a subalgebra of $A$ stable by holomorphic calculus.
A matrix of elements in $A$ is also said to be $p$ summable if all its coefficients are $p$ summable.

These formulas still hold in $\KO$ theory as the proof of Connes~\cite{c85} transfers directly in this setting.
Assume form now on that $A$ is trivially graded.
   \begin{prop}
           Let $e$ be a projector in $\M_k( A_\C)$ together with $J_1$ an anti-unitary as in~\cref{finiterealsec} and $(H,\phi, \epsilon, F,J_2)$ $p$-summable for an even Fredholm module $[H, \phi, F, J_2]$ on $A$ with the same $\KO$ dimension as $e$, $F^* = F$ and $F^2 = \id$.
           The index pairing between the class of $(e,J_1)$ and the class of $(H,\phi,\epsilon, F,J_2)$ in $\K$ theory is given for any $2n+ 1 > p $ by the formula with the complex trace $\tr$ on $H^k$:
           \begin{equation}\label{connesformulaeven}
             \langle F , e\rangle = {(-1)}^n \tr(\epsilon \phi(e)[F, \phi(e)]^{2n})
           \end{equation}
           For $u$ a unitary in $\M_k( A_\C)$ and a anti-unitary $J_1$ as in~\cref{finiterealsec} and  $(H,\phi,F,J_2)$  $p$ summable with respect to an odd Fredholm module $[H, \phi, F]$ on $A$ of the same $\KO$ dimension, with $F^*= F$ and $F^2 = \id$.
          The index pairing between the class of $u$ and the class of $(H,F)$ in $\K$ theory is given if  $2n >p$ by:
           \begin{equation}\label{connesformulaodd}
      \langle F , u \rangle = \frac{{(-1)}^n}{2^{2n+1}} \tr(\phi(u-1)^*[F,\phi(u-1)]([F, \phi(u-1)^*][F, \phi(u-1)])^{n})
           \end{equation}
         \end{prop}
         \begin{remark}
           Same formulas gives the pairing between $\KO^n( A)$ and $\KO_{n+4}( A)$ taking care of dividing the even integer one obtain by 2.
           This is the content of the following diagram:
           \begin{equation*}
              \begin{tikzcd}[row sep = small, column sep = tiny]
                 \KO_{n+4}( A) \arrow[r]\arrow[d, "\bm{\alpha}"]       &
                 \KU_{n}( A) \arrow[r]\arrow[d, "\bm{\alpha}"]			&
                 \KO_{n+2}( A_\C)  \arrow[d, "\bm{\alpha}_\C"]       \\
                 \Z \arrow[r]                                       &
                 \Z \arrow[r]			                                &
                 \ZZ
               \end{tikzcd}
             \end{equation*}
         \end{remark}
        \begin{remark}
          These formulas forget the antilinear operator $J$ when $\KKO$ elements are given as complex Kasparov triples with real structure.
Actually such formulas could have been derived from the following commutative diagram:

           \begin{equation*}
              \begin{tikzcd}[row sep = small, column sep = tiny]
                 \KO_{n-1}(A) \arrow[r]\arrow[d, "\bm{\alpha}"]       &
                 \KO_{n}(A) \arrow[r]\arrow[d, "\bm{\alpha}"]			&
                 \KU_{n}(A_\C)  \arrow[d, "\bm{\alpha}_\C"]       \\
                 0 \arrow[r]                                       &
                 \Z \arrow[r]			                                &
                 \Z
               \end{tikzcd}
                \end{equation*}
          Such a diagram exactly say that the antilinear operator is forgotten.
          The presence of such an operator can however gives us some vanishing results if one try to use these formulas to pair $\KO$ elements with incompatible degree.
          For example when trying to pair and $[u, J_1] \in \KO_{1}( A)$ and  $[H , \phi, F, J_2] \in \KO^{1\pm 2}( A)$  by Writing $\frac{{(-1)}^n}{2^n} \tr(u^*[F,u]([F, u^*][F, u])^{n})$.
          The antilinear $J = J_1 \otimes J_2$ commutes with $u$ and anti-commutes with $F$ so that
          \begin{flalign*}
 \tr(u^*[F,u]([F, u^*][F, u])^{n})     &=          \tr(JJ^*u^*[F,u]([F, u^*][F, u])^{n})\\
 &= \tr(J^*u^*[F,u]([F, u^*][F, u])^{n}J)\\
 &= -\tr(u^*[F,u]([F, u^*][F, u])^{n})   
          \end{flalign*}
         
          And we obtain $\tr(u^*[F,u]([F, u^*][F, u])^{n}) = 0 $.
          See~\cite{k16} for an exhaustive treatment of such computations.
        \end{remark}

        Here we treat real Fredholm modules as given by modules on algebras tensored by Clifford algebras.
        Following recommendations in~\cite{c85}, in~\cite{k86} is defined in the graded setting, cyclic cohomology via characters of real differential cycles for any graded algebra over a field.
     
         A Fredholm module $(H , \phi, F)$ over $ A \gotimes \clifpq$ gives an $n$ differential cycle $(\Omega, \rho, d, \int)$ over $\mathcal A_F^p$ defined for $p\leq n$ by
         \begin{itemize}
         \item $\Omega_m$ is the norm closure of operators of shape $\phi(a_0)da_1 \cdots da_m = \phi(a_0)[F, \phi(a_1) ]\cdots [F, \phi(a_m)]$  for $a_i $ all in $\summable^p_F$
         \item $\rho = \phi : \summable^p_F \to \Omega_0$
         \item $d \left(\phi(a_0)[F, \phi(a_1) ]\cdots [F, \phi(a_m)]\right) =[F,\phi(a_0)][F, \phi(a_1) ]\cdots [F, \phi(a_m)]$
           \item $\int : \Omega_n \to \C$ given by  the graded trace on $H$ multiplied by some scalar $\lambda_n$.
         \end{itemize}
Coefficients $\lambda_n$ are given by Connes as to ensure the compatibility of such a Chern characters with the exterior product of cycles/Fredholm modules.

\begin{align}\label{lambda}
  \lambda_{2n} &=  (-1)^nn ! &
  \lambda_{2n+1} &= (-1)^n \Gamma(n+ \frac{3}{2} ) = (-1)^n \frac{\sqrt \pi (2n+1)!}{2^{2n+1} n!}
\end{align}

Note that in Connes'~\cite[IV.1.$\beta$]{c95}, such coefficients are given respectively by $(-1)^n n!$ and $(-1)^n\sqrt{2i} \Gamma(n+ \frac{3}{2} )$.
To have a compatibility with the operator $\sigma$, a generator of $\HC^2(\C)$ given by $\sigma(1,1, 1) =1$; these two sets of coefficients are each fixed upon multiplication by a scalar.
Connes' choice automatically ensure compatibility with exterior products of two even complex Fredholm modules and even for an even and an odd one.
To have the odd/odd compatibility one fix coefficients as in Connes'.
In the real setting, considering only degree zero Fredohlm modules on algebra eventually tensored by a Clifford algebra, the product of two Fredholm modules of degree one gives a degree 2 element, the choice~\cref{lambda} is compatible with exterior product  $\KO^0(A \gotimes \clifpq) \otimes \KO^0(A' \gotimes \clif[p',q']) \to \KO^0(A \gotimes A' \gotimes \clif[p+p',q+q'])$..
The extra $\sqrt{2i}$ in Connes' comes from the isomorphism of $\cclif[2,0]$ with $\cclif[1,1]$ given by sending a generator $\epsilon$ to $i\epsilon$ followed by the Morita invariance isomorphism in cyclic homology for $\M_2(\C)$.
Note that, in the real setting of Fredholm modules with symmetries, one may define eight sets of such coefficients in a compatible way as to ensure compatibility for Chern character of real Fredholm modules with exterior product. 
This will be detailed elsewhere.

          Now, in~\cite{k16} is defined the pairing between Karoubi/Van Daele $\K$ theory and any class giving a cyclic homology element through characters of Connes' cycles $\left(\Omega, \rho, d , \int\right)$ over the algebra $ \mathcal A \gotimes \clif[q,p]$ for $\mathcal A$  some trivially graded real Banach $*-$algebra.
          Its pairing with a class $[s_1, s_2]$  in $\kok_{p-q}(\mathcal A)$ is given by
          \begin{equation*}
            \mu_n \int \tr(\rho(\epsilon)\rho(s_1 - s_2)(d(s_1-s_2))^n)
          \end{equation*}
Coefficients $\mu_n$ are then defined as to give index formulas as the equality $\left<F, x\right> = \left<ch(F), x \right>$:

\begin{align*}
  \mu_{2n} &= \frac{1}{2^{2n+1} n !} &
  \mu_{2n+1} &= \frac{1}{2^{2n+2} \Gamma(n +\frac{3}{2} )} = \frac{(-1)^n n!}{2 \sqrt \pi (2n+1 )!}
\end{align*}

Connes formulas given in~\cref{connesformulaeven,connesformulaodd} can be seen as a pairing between the Chern class of a summable Fredholm module written as $[H \gotimes R, F \gotimes \id] $ or $[H \otimes R, F \otimes \omega] $ and Karoubi's $\kok$ theory given by elements $[(2p-1) \otimes \id]$ and $[ \left(
  \begin{smallmatrix}
    0 & u^* \\ u & 0
  \end{smallmatrix}\right) \gotimes \id]$ recalling the result from \cite{k86}
\begin{prop}
          Let $\clif$ be a real Clifford algebra and $\tr_{\clif}$ the graded trace on $\clif$ defined as 1 on one product of all Clifford generators and 0 on the product of less generators.
          The exterior product with $\tr_{\clif}$ gives an isomorphism between $\HCO^*(\mathcal A) $ and $\HCO^*(\mathcal A \gotimes \clif)$
        \end{prop}

The Chern Character $s \in \HCO^1(\func_\infty(\interoo{-1, 1}^{0,1})\otimes \clif[0,1])$ of the class $\bm \alpha_{0,1}$ is given by
\[s(a, b) = \frac{1}{2i\sqrt{\pi}}\int_{-1}^1 \tr_1 (a db e)\] With $\tr_1(M_1 \otimes 1 + M_2 \otimes e) = \tr(M_2)$.
\[s(a, b) = \frac{1}{\sqrt{8\pi}i}\int_{-1}^1 \tr_1 (a db \otimes e)\] Recall that $\tr_1(M_1 \otimes 1 + M_2 \otimes e) = \tr(M_2)$.

Exterior product with a differential cycle $\Omega = (\Omega, \rho, \partial, \int)$ over $\mathcal A$ gives a differential cycle $s \Omega$ over $S^{0,1} \mathcal A$ by suspension.
Let the element $(\tilde s_1, \tilde s_2)$ defined as before in~\cref{fhc} for an homotopy representing the vanishing of $(s_1, s_2)_\C$, and giving a class in $\KO(S^{0,1}  \mathcal A)$. we have the pairing given by the following where we denoted $\tilde s_{i, t}$ for $\rho(\tilde s_{i, t})$:
\begin{align}\nonumber
  \left<\Omega, (\tilde s_1, \tilde s_2)\right> =& \frac{(n+1)\lambda_n\mu_{n+1}}{2i\sqrt{\pi}}\int_{-1}^1 \int \tr \gotimes \tr_1((\tilde s_{1} - \tilde s_{2})(\partial \tilde s_{1} - \partial \tilde s_{2})^n(d\tilde s_{1} - d\tilde s_{2})e)  \\\label{phc}
  =&
 \frac{(n+1)\lambda_n\mu_{n+1}}{i\sqrt{\pi}}\int_0^1 \int \tr((\tilde s_{1} - \tilde s_{2})(\partial \tilde s_{1} - \partial \tilde s_{2})^n(d\tilde s_{1} - d\tilde s_{2}))  
\end{align}
Because $\tilde s_{i, -t} = \overbar {\tilde s_{i, t}}$ and the integral is real valued.

A differential cycle $\Omega = (\Omega, \rho, \partial, \int)$ over $\mathcal  A_\C g\otimes $, gives a differential cycle $s \Omega$ over $S \mathcal A_\C \gotimes \cclif[1]$.
Assume that the pairing of $\Omega$ with the complexification of a given $\KO$ cocycle $(s_1, s_2)$ gives zero.
Assuming furthermore the vanishing of $\eta(s_1, s_2)$, we can pair $s\Omega$ with any cocycle $(\tilde s_1, \tilde s_2)$ over $S\mathcal A_\C \gotimes \cclif[1]$ representing as in~\cref{fheta} the vanishing of $\eta(s_1, s_2)$:
\begin{align}\nonumber
  \left<\Omega, (\tilde s_1, \tilde s_2)\right> =& \frac{(n+1)\lambda_n\mu_{n+1}}{2\sqrt{\pi}i}\int_{-1}^1 \int \tr\gotimes \tr_1((\tilde s_{1} - \tilde s_{2})(\partial \tilde s_{1} - \partial \tilde s_{2})^n(d\tilde s_{1} - d\tilde s_{2})e)  \\ \label{pheta}
  =&
 \frac{(n+1)\lambda_n\mu_{n+1}}{2i\sqrt{\pi}}\int_0^1 \int \tr\gotimes \tr_1((\tilde s_{1} - \tilde s_{2})(\partial \tilde s_{1} - \partial \tilde s_{2})^n(d\tilde s_{1} - d\tilde s_{2})e)  
\end{align}
As if $t \leq 0$, we have $\tilde s_{1,t} - \tilde s_{2,t}= \cos_{\pi t/2}(s_1 - s_2) \gotimes \id$ and $\int \tr(( s_1 -  s_2)(\partial s_{1,t} - \partial s_{2,t})^{n+1})$ is the pairing between $(s_1, s_2)$ and $\Omega$.
This number is assumed to vanish.

Specializing to the Chern character of a Fredholm module $(H, \phi, F)$ on $A$ and to $\mathcal A = \mathcal A_F^p$, we obtain:

  \begin{prop}
    For any class $x = [e] - [1_m] $ in $\KO_{2d}({A})$ given by $p$ summable matrices for some Fredholm module $(H, \phi, F)$ on ${A}$ of $\KO$ dimension $2d-1$.
    If $\{y_t\}_{t \in \intercc{0,1}}$  is a smooth homotopy of $p$ summable projectors between $e_\C$ and $1_{m,\C}$ then for any $2n\geq p$ 
    the $\ZZ$ index pairing between $x$ and $[H, \phi, F]$ is given by the following formula
    \begin{align*}
 \langle [H, \phi, F], x \rangle_2
                                   &= \frac{i(2n+1) !}{2^{4n+3} (n!)^2} \int_0^1 \tr(y_ty'_t[F, y_t]^{2n+1}) dt \mod 2
    \end{align*}
      For any class $[u] $ in $\KO_{2d+1}({A})$ given by a $p$ summable matrix for some Fredholm module $(H, \epsilon, \phi, F)$ on ${A}$ of $\KO$ dimension $2d$.
    If $\{y_t\}_{t \in \intercc{0,1}}$  is a smooth homotopy of $p$ summable unitaries between $u_\C$ and $\id$ then for any $2n\geq p$ 
    the $\ZZ$ index pairing between $[u]$ and $[H, \epsilon, \phi, F]$ is given by the following formula
    \begin{align*}
 \langle [H, \epsilon, \phi, F], x \rangle_2 &= \frac{i(n!)^2}{4\pi( 2n) !} \int_0^1 \tr(\epsilon(y^*_t-1)y'_t([F, y^*_t-1] [F, y_t-1])^{n}) dt \mod 2 
    \end{align*}
  \end{prop}

  Now, the vanishing of $x_C = [e_\C] - [f_\C]$ can also be given by an unitary $u$ in $\U_n(A_\C)$ such that $u e_\C u^* = f$.
  In this case we obtain
  \begin{align*}
     \langle [H, \phi, F], x \rangle_2 &= \frac{1}{2^{2n+2}}  \tr((u^*-1)[F, u - 1] ([F, u^*-1][F, u-1])^{n}) \mod 2 
  \end{align*}

  Now for elements in the kernel of $\eta$, we have:

  \begin{prop}
    For any class $x = [e]$ in $\KO_{2d}({A})$ given by $p$ summable matrices for some Fredholm module $(H, \epsilon, \phi, F)$ on ${A}$ of $\KO$ dimension $2d-2$.
    If $\{y_t\}_{t \in \intercc{0,1}}$  is a smooth homotopy of $p$ summable matrices between $1 - 2e$ and the identity matrix verifying reality conditions of~\cref{finitereal8}, then for any $2n\geq p$ 
    the $\ZZ$ index pairing between $x$ and $[H, \phi, F]$ is given by the following formula
    \begin{align*}
 \langle [H, \phi, F], x \rangle_2
                                   &= \frac{i(2n+1) !}{2^{4n+4} (n!)^2} \int_0^1 \tr(\epsilon y_ty'_t[F, y_t]^{2n+1}) dt \mod 2
    \end{align*}
      For any class $[u] $ in $\KO_{2d+1}({A})$ given by a $p$ summable matrix for some Fredholm module $(H, \phi, F)$ on ${A}$ of $\KO$ dimension $2d-1$.
    If $\{y_t\}_{t \in \intercc{0,1}}$  is a smooth homotopy of $p$ summable matrices between $\frac{1}{2}\left(
      \begin{smallmatrix}
        1 & -iu^* \\ iu & 1
      \end{smallmatrix}\right)$ and $
    \frac{1}{2}\left(
      \begin{smallmatrix}
        1 & -i \\ i & 1
      \end{smallmatrix}\right)$ verifying the reality conditions then for any $2n\geq p$ 
    the $\ZZ$ index pairing between $[u]$ and $[H, \epsilon, \phi, F]$ is given by the following formula
    \begin{align*}
 \langle [H, \epsilon, \phi, F], x \rangle_2 &= \frac{i(n!)^2}{2\pi( 2n) !} \int_0^1 \tr(\epsilon(y^*_t-1)y'_t([F, y^*_t-1] [F, y_t-1])^{n}) dt \mod 2 
    \end{align*}
  \end{prop}

  \begin{remark}
    We have given the long exact sequence in bivariant $\K$ theory, in fact one can give similar formulas for a pairing of an element of $\ker (\eta)$ in $\KO$ homology.
    We then replace in the formula above $y_t$ by a constant $y$ and $y_t'$ by $[F_t', y]$.
    Going one step further, this procedure can even give pairings between two torsion elements.
    For example a formula with a double integral given by homotopies will give the pairing between an element of $\ker(c)$ in homology and an element of $ \ker(c)$ in $\KO$ theory.
    Note that the $\ZZ$ index pairing of an element of $\ker(\eta)$ and an element of $\ker(c) \cup \ker(\eta)$ always vanishes.

    The problem of finding formulas as for the $\ZZ$ index of skew adjoint elliptic operators is poorly unresolved with such a treatment.
    A real skew adjoint elliptic operator $D$ on a compact smooth manifold $M$ gives a class in $[D] $ in $\KO^{-1}(\fun(M))$ for which the pairing with the unit of the ring $\KO_0(\fun(M))$ gives the dimension of the kernel of  $D$ mod 2.
    Now $iD_\C$ does not necessarily gives a trivial class in $\KU^1(D)$ and interesting refinements of real $\KO$ theory over the complex $\KU$ disappears.
  \end{remark}

  \begin{remark}
    For some algebras, we cannot compute any $\ZZ$ pairing this way.
    For example for  ${A}$ the algebra of functions on the space obtained by identifying the boundary of a 2 dimensional disk by a rotation of even order $p$ :
    ${A}  = \fun(D^2/ \Z_p)$
    Using the long exact sequence for the inclusion of the circle $S^1$
  \begin{equation}\label{LESRP2}
    \KO_{i+2} \totxt{\times 2} \KO_{i+2} \to \tilde \KO_i({A}) \to \KO_{i+1} \totxt{\times 2} \KO_{i+1} \to \cdots
  \end{equation}
  Combining with the exact sequence of~\cref{kwprop},
  The reduced $\KO$ and $\KU$ theory/homology of this algebra then appears to be of pure torsion.
  The action of complexification, realification and $\eta$ on them can also be recovered from the exact sequence.
   Non zero $\ZZ$ pairings are given by
   $\left<\KO_0, \KO^7 \right>$,
   $\left<\KO_0, \KO^6 \right>$,
   $\left<\KO_1, \KO^7 \right>$,
   $\left<\KO_6, \KO^5 \right>$,
   $\left<\KO_6, \KO^4 \right>$ and
   $\left<\KO_7, \KO^5 \right>$.
   None of it can be computed by our method as the involved elements do not fall in the kernel of $\eta$ or complexification.
   In fact, as every $\KU$ and $ \KO$ groups are torsion, integer valued pairings always vanish and we cannot obtain our $\ZZ$ index as an integer modulo 2.
  \end{remark}
\subsection{Real structures, and quantum symmetries}

  In~\cref{secrealkko} we changed $J$ for $J\epsilon$ or $F$ to $iF$ to obtain different relations between the antilinear operator, the Fredholm one and the grading.
  A more natural way would be to take the whole group $G$ generated by $J$ and eventually $c$ and to keep track of commutation or anti-commutation relations with the Fredholm operator.
  Following~\cite{fm13} we define:
  \begin{defi}
    We call an extended symmetry group the data $(\gct, \texttt t, \texttt c) $ of a group $\gct$ and two group morphisms $\texttt t: \gct \to \ZZ$ and $\texttt c: \gct \to \ZZ$.

  \end{defi}
  Define for such a group $\gct = (\gct, \texttt t, \texttt c) $ their representation on a graded Hilbert space by imposing which elements are even and which one are anti-unitary:
  \begin{defi}
    A representation of  $(\gct, \texttt t, \texttt c)$ on a graded Hilbert space $(\hilb, \epsilon)$ is given by a group morphism $\rho$ from $ \gct$  to the group of unitaries and anti-unitaries of    $\hilb$
    such that for any $g \in \gct$,
    \begin{itemize}
      \item $\rho(g)$ is unitary if and only if $\texttt t(g) = 0$
      \item $\rho(g) \epsilon = {(-1)}^{\texttt c(g)} \epsilon \rho(g)$
    \end{itemize}

  \end{defi}

  This definition extends to Kasparov triples.
  Let $A_\C$ and  $B_\C$ be two trivially graded complex $\cstars$ with real structures:

  \begin{defi}
    An $A_\C - B_\C$ $\gct$ Kasparov bimodule is the data
    $(E, \epsilon, \phi, F, \rho)$ where  $(E, \epsilon)$ is a graded Hilbert $B_\C-$module,  $\phi: A_\C \to \endo_B(E)$ is an even $*-$morphism, $F \in \endo_B(E)$ is odd with respect to the grading $\epsilon$ and self adjoint and of squarer $\id$ modulo compact operators    $\rho $ is a group morphism between $\gct$ and the group of unitaries and anti-unitaries of $E$.
    Furthermore for any $g \in\gct$ the operator $\rho(g)$ verifies the following
    \begin{itemize}
      \item  $\rho(g)$ has degree $\texttt c(g)$ with respect to the grading and graded commutes with $F$ modulo compact operators
      \item For every $a \in A_\C$ and $b \in B_\C$:
      \begin{align*} 
        \rho(g)\phi(a) = \phi(a) \rho(g) &&\rho(g)b = b \rho(g) && \text{if }\texttt t (g) = 0\\
        \rho(g)\phi(a) = \phi(\bar a) \rho(g) &&\rho(g)b = \bar b \rho(g) && \text{if } \texttt t (g) = 1
    \end{align*}
    \end{itemize}
  \end{defi}

This can be used to extract invariants of topological insulators as we discuss now.
The model for the dynamic of our system is described by a first quantized Hamiltonian $H$, a self-adjoint operator acting on a complex Hilbert space $\hilb$.
We neglect the interaction between electrons so that this Hilbert space represents the state space of one electron.
The Fermi energy is a real number $\mu$ describing the statistics of our system.
At zero temperature and with $\mu$ is in a spectral gap of $H$, our electrons occupy all energies lower than $\mu$.
We can then consider the space of such energy levels by looking at the projection of $H$ below $\mu$: $P_{\leq \mu}(H)$.
Assuming some finiteness condition we can then interpret $P_{\leq \mu}(H)$ as living in some $\KU$ group of an algebra.
This algebra must at least contain $H$, or just the compact calculus of $H$ is $H$ is unbounded.
We want to extract the topological property of our Hamiltonian, that is the properties stables after some perturbation of $H$ by a small operator.
This perturbation operator must still live among a given class of plausible Hamiltonians.
The definition of such lead to the algebra $A$.

It has been argued that this algebra can be made non-commutative to represent the uncontrollable disorder.
The model on lattices $\fun(\Omega)\rtimes \Z^d$ first considered by Bellissard and al. in~\cite{bvs94}. See the monograph~\cite{ps16} for a presentation of such a model and its $\K$ theory. 
 Models building out of Roe algebras first studied by Kubota~\cite{k17b} with the reduced $C_{Roe}^*(\Z^d)$ and  Meyer~\cite{em18a} with an interpretation using the uniform Roe algebra.

 Now symmetries of our dynamics is given by any operator acting on the phase space of one electrons.
The phase space being the projective space of $\hilb$ where two vectors are identified if they are each a multiple of the other.
By Wigner's theorem~\cite[2.2]{w95}, symmetries can be represented as operators on $\hilb$ which are either unitaries or antiunitary:

\begin{theo}
  Let $S$ acts on the projective space $\proj(\hilb)$, keeping invariant the quantity $\frac{\|\left< \Phi, \Psi \right>\|}{\|\Phi\|\|\Psi\|}$ for $[\Phi]$, $[\Psi]$ in $\proj(\hilb)$, then there is a linear operator $\tilde S$ on $\hilb$ that lifts the action of $S$.
  Such an operator is either unitary or antiunitary.
  Moreover it is unique up to multiplication by a unitary complex scalar if $\dim \hilb \geq 2$.
\end{theo}
Symmetries of our dynamics  are given by symmetries of the phase space commuting with the dynamics given by the Hamiltonian $H$.
We assume from now on that the condition $\dim \hilb \geq 2$ is satisfied.
To such a symmetry $S$ we associate two elements in $\ZZ$:
\begin{itemize}
\item $\texttt t(S)$ being 1 if $S$ preserves the arrow of time, $-1$ otherwise.
  \item $\texttt c(S)$ being $1$ if $\tilde S$ is unitary, $-1$ if is anti-unitary.
  \end{itemize}
  Then $\texttt t(S) \texttt c(S)$ gives the sign in $H\tilde S = \pm \tilde S H$.

  If now a group $\gct$ is acting as symmetries it gives 
 $\tilde \gct$ that fits the following exact sequence of groups:
  \[1 \to \U_1 \to \gct \totxt{\pi} \gct \to 1\]
  And verifies that for every $ g $ in $\gct$ and every $\lambda $ in $\U_1$  we have $ g \lambda = \begin{cases}
    \lambda g &\mbox{if }\texttt t ( \pi(g)) = 0\\
     \bar \lambda g &\mbox{if }\texttt t ( \pi(g)) = 1
  \end{cases}$
 
    We give call the extended symmetry group $(\tilde \gct, \texttt t \circ \pi, \texttt c \circ \pi)$ a \emph{twisted extension} of $\gct$.
  
  Denote $\ct$ the group $\ZZ \times \ZZ$ with the structure of an extended symmetry group given by projections $\texttt t $ and $\texttt c$ on the first and on the second variable.
  \begin{prop}
  There are up to isomorphism 10 twisted extension of subgroups of $\ct$ as before.
  To such an extension, there exist an integer $d$ such that Kasparov $A-B \gct$ bimodules are in bijective correspondence with Kasparov bimodules of $\KO$ dimension $d$ or Complex Kasparov bimodule of degree $d$.
  \end{prop}
  \begin{proof}
  We distinguish by the subgroup of $\ct$ the twisted extension is build from.
  Taking into account morphisms $\texttt c$ and $\texttt t$ there are fives choices:
  the full $\ct$, the null $\{1 \}$ and three groups isomorphic to $\ZZ$, the kernel $\ct^{\texttt t}$ of $\texttt c$, the kernel $\ct^{\texttt c}$ of $\texttt t$, the kernel $\ct^{\texttt c \texttt t}$ of $\texttt c + \texttt t$.

  Now the isomorphism class of the quantum CT group is unique for $\{1\}$ and  $\ct^{\texttt c}$.
  For each of the two groups $\ct^{\texttt t}$ and $\ct^{\texttt{tc}}$ there is two equivalence classes of quantum CT group.
  Those two classes are characterized by the fact that for $T \in \tilde \sct$ such that $\texttt t (T) =1 \in \ZZ$ then $T^2 = 1$ or $-1$.
  Over the full $\ct$ we have four distinct quantum CT groups.
  With $T$ as before and $C \in \ct$ such that $\texttt c (C) = 1 \in \ZZ$, the four classes are characterized by the sign of $T^2$ and by the sign of $CTC^{-1}T^{-1}$.

  Now we saw in~\cref{secrealkko} that for a real Clifford algebra one can associate uniquely a quantum CT group such that graded complexification of representation of the Clifford algebra are in 1:1 correspondence with representation of the CT group.
  The four even $\KO$ dimensions correspond to twisted extensions of  $\ct^{\texttt t}$ and $\ct^{\texttt{tc}}$.
  The four odd $\KO$ dimensions correspond to  Extensions of  $\ct^{\texttt c}$.
  For complex Clifford algebras the same can be said and even algebras are associated to the trivial group, odd one to the group $\ct^{\texttt c}$.
  \end{proof}

  Now, the operations of $\eta$ complexification and realification can be understood through the scope of the representation of such groups.
  To a morphism between CT groups one build through induced representations a map at the level of Kasparov triples.
  Now: 

  \begin{prop}
    Every map in the exact sequence~\cref{kwprop} can be interpreted in such a way:
    \begin{itemize}
    \item Complexification is given by taking the cokernel map of $\texttt c$.
    \item    Realification is given by the only corresponding inclusion of a group of order 1 and 2.
      \item $\eta$ is given by either the cokernel map of $\texttt t$ or an inclusion
    \end{itemize}
  \end{prop}

  \begin{proof}
    It is obvious for the complexification map.
    The rest is just a reformulation of the interpretation of the functors as some forgetful map.
  \end{proof}

  \begin{remark}
    In fact $\id, \eta$, $\eta^2, r, r\beta, r\beta^2, r\beta^3$ and $c$ are the only non zero morphism in $\KK$ theory obtained by some morphism between CT groups and corresponding induced representations.

    This can be seen as follows.
    Such a functor $\mathrm F$ between $\KKO$ groups gives a map for every one of the eight corresponding CT groups to an other.
    The construction being functorial and the involved groups finite in number one see that the sequence $(\mathrm F^i)_{ i \in \N}$ stabilizes for some $i$ big enough.
    Now, $\id$, $\eta$ and $\eta^2 $ are the only elements of $\KO_*$ that stabilizes.
    The same can be said for complex $\K$ theory, in which case the only operator we obtain is $\id$.
    If $\mathrm F$ is such a functor between different flavors $\KKO$ and $\KKU$, composition with $r$ or $c$ gives the statement.
  \end{remark}

  \begin{remark}
    This can be extended to any $\psi-$standard extended symmetry group of~\cite{fm13}.
  \end{remark}

\subsection{Complex structure and spin operators}

A lift of a real element in a $\KO$ group to a $\KU$ class can be interpreted as the existence of a spin operator commuting with the symmetry:
\begin{defi}\label{finitespinop}
  A spin operator for a class $[s_1, s_2, J]$ is any operator $S$ verifying:
  \begin{align*}
    Ss_1 = s_1 S && Ss_2 = s_2 S && SJ = -JS &&S^2 = 1
  \end{align*}
\end{defi}
For Kasparov triples, this adapts to give:
    \begin{defi}\label{spinop}
      Let $(E, \phi, F, J)$ a Kasparov triple of $\KO$ dimension  $d$ on $A_\C- B_\C$.
      We call $S \in \endo_{B_\C}(E)$ a \emph{spin operator} on $(E, \phi, F, J)$ if
      \begin{align*}
        [\phi(A), S] &= 0& FS &= SF& S^2&= 1& JS &= -SJ 
      \end{align*}
    \end{defi}

    This operator $S$ induces a grading on $E$ for which $J$ become an antilinear isomorphism between the odd and the even part as in our discussion above on realification of complex Kasparov triples.
    
 It gives a lift in $\KKU(A_\C, B_\C)$ for the realification functor as on the real locus  $E_\R = \{e\in E, Je = e \}$, we define the action of the imaginary as $I = iS$.
    We rephrase a part of the exact sequence $\ref{kw}$ by :
  \begin{prop}
    Let $\bm x \in \KKO_n(A,B)$ then there is a Kasparov module with real structure representing $\bm x$ admitting a spin operator acting on it if and only if $\eta \bm x  = 0$.  
  \end{prop}

  For a periodic two dimensional system of fermions without any disorder, the Fermi projection of a time reversal and translation invariant tight binded Hamiltonian gives a class in $\KO_4(A)$ for $A$ the algebra $C^*(\Z^2)$.
  This group is isomorphic to $Z + \ZZ$ as can be shown using twice the split long exact sequence associated to the inclusion of a point in the circle.
    The free part represents the number of pairs of filled bands.
    The torsion part gives the Kane-Mele invariant of~\cite{km05} as a pairing between the $\KO$ class in $\KO_4(A)$ with the fundamental class of the 2-dimensional torus in $\KO^2(A)$.
    See~\cite{gs16} for an explainaition based on $T-$duality.

  The action of $\eta$ on this $\KO$ element gives zero.
  There is then a complex lift in $\KU_4(A_\C)$.
  We can the represent our $\KO$ class as a symmetry admitting a spin operator $S$ commuting it.
  This allows us to define as in~\cite{p09} a spin Chern number.
  The operator $S$ gives a splitting of the $\KO$ class $x = x^+ + x^-$ as in~\cref{realificationfinite}.
  The spin Chern number is then the Chern number associated to $x^+$ or $x^-$.
  It is then the Chern number of the lifted class in $\KU_0(A_\C)$.
  We then recover the result of~\cite{s13a} that this spin Chern number gives the Kane-Mele invariant after reduction modulo 2.
    As explained in~\cite{p09}, a spin Chern number may be defined even without commutativity if the spin operator with our symmetry: when the spectral flattening of $\frac{1}{2}(S + sSs)$ is permitted, this operator defines a commuting antilinear operator and a spin Chern number.

  In an other language, and more generally, as was shown in~\cite{hmt18} the Kane-Mele invariant can be expressed as a second Stiefel Whitney class on the 2-torus.
  The equality modulo 2 of the complex and real invariants as discussed above could have been given using results from the theory of characteristic classes

  \begin{prop}
    Let $E$ be a complex vector bundle on a compact space $X$.
    Denote $X_\R$ its underlying real vector bundle.
    We have an equality in $H^2(X, \Z_2)$ of the second Stiefel Whitney class of $E_\R$  and the first Chern class modulo 2 of $E$
  \end{prop}

\subsection{\texorpdfstring{$\eta$}{eta} as dimensional reduction}

We give in this last paragraph an other interpretation of operation $\eta$.
In~\cite{rsf10} it is shown how a $d-$dimensional topological insulator lead to a $d-1$ topological insulator by a procedure called \emph{dimensional reduction}.
In fact, such a reduction can be interpreted as a product with $\eta$.
Recall that the $\cstarr$ of the group $\Z$ is by Fourier transform isomorphic to $\fun(S^{0,1})$,  $\KKO$ equivalent to a sum $\func(\R^{0,1}) + \R$.
There is a restriction map $\KO(\fun^*(\Z)) \to \KO(\func(\R^{0,1})$ given by sending the class of a projector $p$ to the difference of $[p]$ with some trivial projector.
The trivial projector can be seen as an application of the pullback map in $\K$ theory for the map sending $S^{0,1}$ to a real point.
Calling this point $\infty$, this map  is the Kasparov product with $\left[\func(\R^{0,1}) + \func(\R^{0,1}) , (id, ev_\infty), \left(
  \begin{smallmatrix}
    0 & 1 \\ 1 & 0
  \end{smallmatrix} \right) \right]$ in $\KKO(\fun(S^{0,1}), \func(\R^{0,1}))$.

\begin{prop}
  At the level of $\KO$ theory, the dimensional reduction given by $\fun^*(\Z) \to \R$ identifies in $\KO$ theory as application of the previous restriction map followed by application of $\eta \in \KKO(\func(\R^{0,1}), \R)$.
\end{prop}

At the level of Bloch Hamiltonian, representing coordinats in the Bloch-Fourier space by  $(k_1, \cdots, k_d)$, this dimesnional reduction amounts to:
\[H(k_1, \cdots , k_d) \mapsto H(k_1, \cdots, k_{d-1}, 0)\]

At the level of tight binding Hamiltonians, when translation by $i$ in the $j-th$ coordinate in $\summable^2(\Z^d)$ is written $t^i_j$, the dimesnional reduction is given by:
\[H = \sum_{i_1, \cdots, i_d} \alpha_{i_1 , \cdots, i_d} t_1^{i_1} \cdots  t_d^{i_d}\mapsto \sum_{i_1, \cdots, i_d} \alpha_{i_1 , \cdots, i_d} t_1^{i_1} \cdots  t_{d-1}^{i_{d-1}}\]

 \printbibliography
\end{document}